\documentclass[12pt,twoside,a4paper]{amsart}

\usepackage{amsaddr}
\usepackage[margin=1.2in]{geometry}
 \usepackage{setspace}
\usepackage{mathrsfs}
\usepackage{amssymb,amsmath,amsfonts,amsthm,color}
\usepackage{thmtools,thm-restate}
\usepackage[latin1]{inputenc}
\usepackage{enumerate}
\usepackage{setspace,graphicx,setspace,multirow,lscape,longtable,threeparttable,dcolumn}
\usepackage{proba}
\usepackage{booktabs}
\usepackage[round]{natbib}
\usepackage{mathtools}
\usepackage{dsfont}
\usepackage{placeins}
\usepackage{tikz}
\usepackage{caption}
\usepackage{subcaption}
\usepackage{graphicx}
\graphicspath{ {./Figures/} }

\usepackage{hyperref}
\hypersetup{
    colorlinks,
    linkcolor={red!50!black},
    citecolor={blue!50!black},
    urlcolor={blue!80!black}
}

\newtheorem{theorem}{Theorem}
\newtheorem*{theorem*}{Theorem}

\newtheorem{assumption}{Assumption}
\newtheorem{lemma}{Lemma}

\newtheorem{corollary}{Corollary}

\newtheorem{assumptionalt}{Assumption}[assumption]

\newenvironment{assumptionp}[1]{
  \renewcommand\theassumptionalt{#1}
  \assumptionalt
}{\endassumptionalt}

\def\wl{\par \vspace{\baselineskip}}

\def\m{\mathcal}

\def\C{\mathbb{C}}

\def\E{\mathbb{E}}
\def\R{\mathbb{R}}

\def\to{\rightarrow}
\def\1{\mathds{1}}
\def\ml_#1{\limits_{\mathclap{\substack{#1}}}}
\def\ls{\lesssim}
\def\gs{\gtrsim}
\def\t{\widetilde}
\def\h{\widehat}
\def\s{\mathscr}
\def\co{\mathsf{co}}

\title[Distributional Counterfactual Analysis]{Distributional Counterfactual Analysis in High-Dimensional Setup}
\author{Ricardo Masini}
\address{\vspace{-0.425cm}University of California, Davis\\\textnormal{Department of Statistics}}
\email{rmasini@ucdavis.edu}
\date{September 2023}
\thanks{I would like to express my gratitude to Marcelo Medeiros,  Matias Cattaneo, and William Underwood for fruitful technical discussion and guidance}

\begin{document}

\begin{abstract}
\noindent In the context of treatment effect estimation,  this paper proposes a new methodology to recover the counterfactual distribution when there is a single (or a few) treated unit and possibly a high-dimensional number of potential controls observed in a panel structure. The methodology accommodates, \emph{albeit} does not require, the number of units to be larger than the number of time periods (high-dimensional setup). As opposed to modeling only the conditional mean, we propose to model the entire conditional quantile function (CQF) without intervention and estimate it using the pre-intervention period by a $\ell_1$-penalized regression. 
We derive non-asymptotic bounds for the estimated CQF valid uniformly over the quantiles. The bounds are explicit in terms of the number of time periods, the number of control units, the weak dependence coefficient ($\beta$-mixing), and the tail decay of the random variables. The results allow practitioners to re-construct the entire counterfactual distribution. Moreover, we bound the probability coverage of this estimated CQF, which can be used to construct valid confidence intervals for the (possibly random) treatment effect for every post-intervention period. We also propose a new hypothesis test for the sharp null of no-effect based on the $\m{L}_p$ norm of deviation of the estimated CQF to the population one. Interestingly,  the null distribution is quasi-pivotal in the sense that it only depends on the estimated CQF,  $\m{L}_p$ norm,  and the number of post-intervention periods,  but not on the size of the post-intervention period. For that reason,  critical values can then be easily simulated. We illustrate the methodology by revisiting the empirical study in \citet*{acemoglu2016value}.
\wl
\noindent
\wl
\end{abstract}

\maketitle

\newpage

\section{Introduction}\label{S:intro}

Treatment effect estimation has always been an active research topic in Economics. A causal statement in empirical studies usually concerns the effects of a given treatment  (policy, intervention, event) on the population of interest. Since each unit in the population is \emph{either} treated or untreated at a given time,  one of the outcomes is unobservable. Therefore,  causal statements often rely on constructing \emph{counterfactuals} based on the outcomes of the treated and untreated units (controls). Notwithstanding, definitive cause-and-effect statements are usually problematic to formulate, given economists' constraints in finding sources of exogenous variation. For a comprehensive account in the context of program evaluation, see \citet*{Abadie2018EconometricMF} and references therein. Recently,  there has been a renewed interest in causal inference using panel (longitudinal) data,  particularly when one observes only a single (or a few) treated unit and several potential controls across time. Most of this development has focused on the mean treatment effect. However, much more can be learned by taking a more comprehensive approach. For that reason, we propose a new methodology to evaluate the impact of interventions on the entire \emph{distribution} of the variable of interest. 

The setup consists of a panel structure where a single unit is treated at a given time, and only a few (possibly one) post-intervention periods are available.  Although not required, the methodology accommodates situations when the number of units is much larger than the number of time periods (high-dimensional setup). As previously mentioned, we depart from the typical approach of modeling only the conditional mean as in \citet*{cHhsCskW2012}. Instead,  we simultaneously model all the conditional (on the untreated units) quantiles of the unit of interest, allowing us to re-construct the conditional distribution without intervention. 

The benefits of this approach are threefold. First, we have a comprehensive picture of the intervention effect as we have available the (estimated) counterfactual distribution and not only its mean. For instance,  we can detect mean-preserving treatments, such as those that affect only the distribution's tails or variance. Second, although we might take the median as our point estimate,  we have readily available confidence intervals for the (possibly random) treatment effect at each post-intervention period. Third,  it allows us to construct a new hypothesis test for the null of the no-distributional effect. The test procedure is based on a quasi-pivotal test statistic whose critical values can be easily simulated.

Specifically,  we estimate the conditional quantile function using a $\ell_1$-penalized linear quantile regression to account for a potentially large number of units using the pre-intervention period. We provide non-asymptotic probabilistic bounds for the deviation of the estimated quantile function from the true one and the coverage of the estimated conditional quantile in terms of the true one. Both results are valid uniformly over any quantile in a  compact set. The bounds are expressed explicitly in terms of the pre-intervention period, the number of peers, the tail condition, the number of relevant peers, and the measure of serial dependency. This makes evident the trade-off between these quantities in the convergence rate. As a corollary, we propose confidence intervals for the treatment effect, valid under simple rate conditions. 

Furthermore,  the proposed test for the null hypothesis of the no-distribution effect is based on the $\m{L}_p$-norm of the conditional quantile function before and after the intervention without relying on permutation tests. Interestingly,  the test statistic has a quasi-pivotal distribution under the null, in the sense that it only depends on a nuisance parameter that can be estimated using the pre-intervention period and, therefore, is valid even for a single post-intervention period. This test's critical values or p-values can be efficiently obtained via a simple Monte Carlo simulation that only draws from independent standard uniform random variables.  

The literature on treatment effects is quite vast and diverse. Nonetheless, we see our methodology based on three key features: distributional effects, few treat units,  and high-dimensional setup. We borrow core ideas from the quantile treatment effect literature for the first one; see, for instance, \citet*{vCcH2005}.  However,  we do not exploit heterogeneity among individuals as the source of our distributional effect; this approach was recently taken in the Distributional synthetic controls (\cite{gunsilius2021distributional}). Instead,  we apply this idea in a panel structure assuming a \emph{single} treated unit. In other words, we only observe a single realization of a treated unit at any given time period, which motivates our second design feature. By focusing only on ``few treated units'' techniques,  our setup is closely related to  \citet*{cHhsCskW2012} and  \citet*{cCrMmM2016}, the latter also considers a high-dimensional setup, which leads to our third feature. High-dimensional methods and treatment effects are reviewed in \citet*{BCH2014} with the focus on the idea of dealing with treatment effect estimation with a high-dimensional control group for the case when the treatment is considered exogenous conditional on a very large number of unaffected (by the treatment) characteristic of the treated unit. In our case, these observable characteristics will be replaced by the control units' outcome.

In all those papers considering a high-dimensional setup, some sparsity assumption is usually evoked to impose a low-dimensional structure. Exploiting a strong factor structure in a high-dimensional environment,  \citet*{fan2021we} proposes a counterfactual of the conditional mean, which can also be seen as a sparsity assumption in terms of the eigenvalue of the design matrix. Finally,  the methodology can also be loosely connected to the Synthetic Control Methods (SCM). For a modern reference, see \citet*{abadie2021using}.  In particular,  \citet*{chen2020distributional} extends the canonical SCM and considers macro-level intervention in a single unit that consists of several individuals (or sub-unit as labeled by the author). In that sense, there are, in effect, several treated (sub-) units at each given period, and distributional effects can be measured as the heterogeneity effect across individuals in the same unit of interest. In a different setup, the same idea of leveraging cross-sectional variation (of a single unit at a given time) is used in \citet*{gunsilius2021distributional}.

The rest of the paper is organized as follows. Section \ref{S:estimator} presents the setup,  defines the estimator, and discusses the main assumptions underlying our results.   In particular,  Subsection \ref{SS:Overview} is dedicated to an overview of the main results, which are formally presented in Section \ref{S:Assumptions_Results},  together with additional technical assumptions. We describe the Hypothesis Testing procedure in Subsection \ref{SS:Hypothesis Testing}, which includes a viable implementation of a simulation procedure to obtain critical values. Section \ref{S:montecarlo} described the Monte Carlo study and brief analysis of the estimator's finite-sample performance.  The proposed methodology is illustrated in an empirical application in Section \ref{S:application}, and Section \ref{S:conclusion} concludes.  All proofs are relegated to the Supplemental Material, which also includes some additional figures from the empirical illustration.

\section{Set-up, Estimator and Main Assumptions}\label{S:estimator}

\subsection{Set-up}\label{SS:definitions}

Suppose we have $n\geq 2$ units (countries, states, municipalities, firms, etc.) indexed by $i\in\{1,\dots,n\}$. For each unit and for every time period $t\in\{1,\ldots, T\}$ for $T\geq 2$, we observe a realization of the random variable $Z_{it}$. Furthermore, we consider that there is \emph{only one} unit that suffers the intervention\footnote{The proposed methodology would still be applicable when there are a few treated units by applying it to each treated unit separately.  That is the case, for instance,  in our empirical application.} (treatment) at time $1< T_0<T$. Without loss of generality, we assume the treated unit to be the unit one ($i=1$). Let $D_{t}$ be a binary variable flagging the periods when the intervention is in place.  Using the potential outcome notation, we can express our observable variable as
\begin{equation*}
Z_{it}= D_{t} Z_{it}^{(1)} + (1-D_{t})Z_{it}^{(0)},
\end{equation*}
where, following the literature on treatment effects, $Z_{it}^{(1)}$ denotes the outcome when the unit $i$ is exposed to the intervention at time $t$ and $Z_{it}^{(0)}$ when it is not.

We are interested in cases where once the unit one suffers the intervention after $t=T_0$ it is kept treated throughout the remaining sample, i.e.,  $D_{s}=1$ implies $D_{t}=1$ for all $t \geq s$.  Hence, we write $D_t=\1\{t\geq T_0\}$.  It is important to stress that even though the intervention variable $D_{t}$ might seem deterministic,  it is, in effect, a random variable to the extent that $T_0$ is a random variable that might depend on both potential outcomes $\{Z_i^{(0)}, Z_i^{(1)}:1\leq i\leq n,  1\leq t\leq T\}$. 

Clearly, in our sample, we do not observe $Z_{1t}^{(0)}$ for $ T_0<t\leq T$; for that reason, we refer to it as the \emph{counterfactual}, i.e., what would the unit of interest has been like had there been no intervention in place. The next assumption restricts the treatment dependency on potential outcomes.  Write $Z^{(0)}_{t}:=(Z^{(0)}_{1t},\dots, Z^{(0)}_{nt})'$ and $Z^{(1)}_{t}:=(Z_{1t}^{(1)},\dots, Z_{nt}^{(1)})'$. We make the following identification assumption

\begin{assumptionp}{A.1}[Identification]\label{A:identification} Suppose that
\begin{enumerate}[(a)]
\item $D_t$ is independent of ${Z}_s^{(0)}$ for all  $t,s\in \{1,\dots, T\}$;
\item $Z_{it}^{(1)} = Z_{it}^{(0)}$ for $i\neq 1$ and $t\in \{1,\dots, T\}$.
\end{enumerate}
\end{assumptionp}

Part (a) requires the treatment (or equivalently, the timing of the treatment $T_0$) to be independent of the potential outcomes under no intervention for all units. We do \emph{not}, however, require the treatment to be independent of the potential outcome under the intervention, namely $Z_t^{(1)}$. Since we are only interested in the treatment effect on the treated, it is analog to the well-known fact in the treatment effect literature that one can consistently estimate the average effect even when $\E(Z_t|D_t)\neq\E(Z_t)$. Part (b) states that the untreated units (peers) are unaffected by the intervention in the unit of interest, or equivalently, its potential outcome is the same with or without the intervention. 

Specifically,  Assumption \ref{A:identification}(a) ensures that,  conditional on the treatment $D := (D_1,\dots, D_T)'$ (or equivalently  on $T_0$),  $Z_1,\dots, Z_{T_0}$ and  $Z_1^{(0)},\dots ,Z_{T_0}^{(0)}$ share the same distribution.
Thus, by postulating an appropriate model for the variables in the absence of intervention 
\begin{equation}\label{eq:model_M_generic}
Z_{1t}^{(0)} = \m{M}(Z^{(0)}_{2t}, \dots, Z^{(0)}_{nt}) + \epsilon_t,
\end{equation}
say,  one can, in principle, estimate it using the pre-intervention sample $\{ Z_1,\dots Z_{T_0}\}$.  Denote the estimated model by $\widehat{\m{M}}$.  Assumption \ref{A:identification}(b) then allows us to extrapolate the (estimated) model to the post-intervention period to construct the desired counterfactual for the unit of interest.  In particular,  it allow us to claim that $\m{M}(Z^{(0)}_{2t}, \dots, Z^{(0)}_{nt})$ has the same distribution as $\m{M}( Z_{2t}, \dots, Z_{nt})$. We then use the post-intervention sample $\{ Z_{T_0+1},\dots Z_{T}\}$ to compute the counter factual as  $\widehat{Z}_{1t}^{(0)}:=\widehat{\m{M}}( Z_{2t}, \dots, Z_{nt})$ for $T_0< t \leq T$.
 
\subsubsection{DGP}
 
I consider the following data generating process (DGP) for the units in the absence of the intervention
\begin{assumptionp}{A.2}[DGP]\label{A:DGP} For $n\geq 2$ and $ T\geq 2$,  the sequence of $n$-dimensional random vector $\{Z_{t}^{(0)}:=(Z_{1t}^{(0)},\dots, Z_{nt}^{(0)})':1 \leq t\leq T\}$ is $\beta$-mixing.
\end{assumptionp}

The GDP described in Assumption \ref{A:DGP} is quite flexible as it leaves the dependency among the $n$ entries of $Z_{t}^{(0)}$ and the serial dependence unspecified. It also does not assume the sequence to be identically distributed or stationary. However, we require the conditional quantiles to be identically distributed as per Assumption \ref{A:CGF} below. The latter is necessary to model the entire (conditional) distribution.

We use absolute regularity ($\beta-$mixing coefficients) as the measure of serial weak dependency of $\{Z_t^{(0)}:1\leq t\leq T\}$. There are several equivalent definitions of $\beta$-mixing coefficients. In our context, the most appropriate appears on page 3 of \citet*{Doukhan1994}, reproduced below for convenience. As opposed to being stated in terms of sigma-algebras as originally defined, the coefficients can be represented as the largest distance (in the total variation norm) between the joint distribution of the past and the future and the product of its marginals. Formally, for $0\leq m <T$, define the $\beta$-mixing coefficient by
\begin{equation}\label{E:beta_mixing_coef}
    \beta_m:=\sup\{\|\P_{(\m{Z}_1^t,\m{Z}_{t+m}^T)} - \P_{\m{Z}_1^t}\otimes \P_{\m{Z}_{t+m}^T}\|_{TV}:1\leq t\leq T\},
\end{equation}
where $\P_{\m{Z}_s^t}$ denote the joint distribution of $(Z_{s}^{(0)},\ldots, Z_t^{(0)})$ for $1\leq s\leq t\leq T$.
Note that $\beta_m$ might depend on $T$ and $n$, but $\{\beta_m\}$ is non-decreasing  sequence in $m$ for each fixed $T$ and $n$. Therefore the process is required to be $\beta$-mixing in Assumption \ref{A:DGP} in in the sense that $\beta_m\to 0$ as $m\to\infty$ for each fixed $T$ and $n$.

Numerous studies have been conducted on the mixing properties of strictly stationary linear processes, encompassing but not restricted to strictly stationary ARMA processes, ``non-causal'' linear processes, and linear random fields. Additionally, research has been developed on related processes like bilinear, ARCH, or GARCH models. For comprehensive insights into the mixing properties of these and other associated processes, refer to Chapter 2 of  \citet*{Doukhan1994}.

To exemplify, Assumption \ref{A:DGP} nests the latent (possibly dynamic) factor model structure by setting
\[Z_{it}^{(0)}=  \mu_i + \lambda_i'  F_t+V_{it},\]
where $\mu_i$ is a (possibly random) fixed effect for unit $i$,  $\lambda_i\in \R^r$; $\lambda_i$ are loading of the $(r\times 1)$ common factor $ F_t$ and $V_{it}$ are the idiosyncratic shocks.

The presence of a common unknown deterministic time-trend $\{\zeta_t:1\leq t\leq T\}$ commonly encountered in the Synthetic Control can also be accommodated. Suppose that
\[Z_{it}^{(0)}=  \zeta_t + U_{it}\]
where $\{U_{t}:=(U_{1t},\dots, U_{nt})':1 \leq t\leq T\}$ fullills Assumption \ref{A:DGP}. Then as long as the conditional quantiles are stable in the sense of Assumption \ref{A:CGF}, the methodology described below can be directly applied. The case of stochastic trends is subtle due to the fact that an infinite sum of a (underlying) mixing process is not necessarily mixing. Hence, even an integrated linear process would require a careful analysis. For conterfactual analysis of the conditional mean involving stochastic trends refer to \cite{masini_jasa}.

The presence of observable (covariates) does not present any particular issue as long as we assume that those covariates are unaffected by the treatment $D_t$. If, for instance, we postulate that
\[Z_{it}^{(0)}= \pi_i'W_{it} +  U_{it},\]
where  $\{U_{t}:1 \leq t\leq T\}$ fullills Assumption \ref{A:DGP} and we carry on the analsys with the ``observable'' $\{\widehat{U}_{t}:1 \leq t\leq T\}$ where $\widehat{U}_{t}:= Z_{it} + \widehat{\pi_i}'W_{it}$ for some estimator $\widehat{\beta}_i$ as long as we are able to control for $\sup_{it}|\widehat{U}_{it} - U_{it}|$.  Refer to Theorem 1 in \citet*{fan2021bridging} for primitive conditions to bound this quantity when $\widehat{\pi}_i$ is the ordinary least squares estimator.

\subsubsection{Distributional Treatment Effect}

Recall that unit 1 is the unit of interest (treated). To ease on the notation we set, for $1\leq t\leq T$,
\[
Y_t := Z_{1t}^{(0)}\qquad X_{t}:=(1,Z_{2t}^{(0)},\dots, Z_{nt}^{(0)})'.
\]
In most counterfactual exercises, the model $\m{M}$ appearing in \eqref{eq:model_M_generic} is the conditional mean or just the linear projection of $Y_t$ onto $X_t$ (see, for example,  \citet*{cHhsCskW2012} and \citet*{cCrMmM2016}). Since the focus is to investigate the potential distributional effects of the intervention of interest, we must model the entire counterfactual distribution. Therefore, we take $\m{M}$ as the collection of the conditional quantiles of $Y_t$ given $X_t$. We could also choose to model the conditional distribution as both fully characterize the distributional effect of the intervention. An interesting comparison between those two approaches can be found in \citet*{koenker2013distributional}. Heuristically,  we measure the distributional treatment effect by the differences it may cause to the conditional quantiles of $Y_t|{X}_t$,  which under Assumption \ref{A:identification} are attributed to treatment on the unit of interest.

Let $F(\cdot|x)$ denote the conditional distribution function of $Y_t$ given $X_t =x$ i.e.,  $F(y|x) := \P(Y_t\leq y|X_t=x)$ for $y\in\R$ and $x$ in the support of $X_t$. Note that under Assumption \ref{A:CGF} below, $F(y|x)$ does \emph{not} depend on $t$. Also, let $Q(\cdot|x)$ denotes the conditional quantile function of $Y_t$ given $X_t$ namely $Q(\tau|x):=\inf\{y\in\R:F(y|x)\geq\tau\}$ for $x$ in the support of $X$ and $\tau\in(0,1)$.   We are left to specify the functional form of the conditional quantiles.

\begin{assumptionp}{A.3}[Linear Conditional Quantile Function]\label{A:CGF} For some non-empty compact  $\m{T}\subset (0,1)$ and every  $\tau\in\m{T}$,  there is a $\theta_0(\tau)\in\R^n$ such that $Q(\tau|x)=x'\theta_0(\tau)$ for every $1\leq t\leq T$.
\end{assumptionp}

Assumption \ref{A:CGF} postulate a correctly specified linear model for the conditional quantile function for all quantiles in $\m{T}$. The consequences of missepecification in the quantile regression model are treated in \citet*{10.2307/3598810}. We could also consider a more flexible specification where we allow the functional form to vary with $\tau$, such that $Q_{Y_t|X_t}(\tau|x)=g_\tau( x,\theta_0(\tau))$ for some know class of function $g_\tau$. Or even a non-parametric specification if, in the empirical application, one expects to have $n$ is much smaller than $T_0$. However, we are interested in covering cases when $n$ might be of the same order or much larger than $T_0$ when a more parsimonious model is desirable.

Also, note that $\theta_0(\tau)$ might not be unique. This lack of identification might come from the fact that the density of $Y_t$ given $X_t$ is not bounded away from zero at the conditional $\tau$-quantile (which is ruled out by Assumption \ref{A:conditional_density} below) or a  rank deficiency of $\E \sum_{t=1}^{T_0} X_tX_t'$ which we allow here.  For each $\tau\in\m{T}$,  we denote by $\Theta^0_\tau \subseteq \R_n$ the set of all vectors that fulfill Assumption \ref{A:CGF}.  It is well known that $\Theta^0_\tau$  is the set of minimizers of the population  objective function $\theta\mapsto \E\rho_\tau(Y_t - X_t'\theta)$ for each $\tau\in
\m{T}$, 
where  $\rho_\tau(z):=z(\tau - 1\{z<0\})$ is known as the check function.  Our model can then be expressed for some compact $\m{T}\subseteq (0,1)$ as
\begin{equation}\label{eq:model}
\m{M}:= \{x \mapsto Q(\tau|x):\tau\in\m{T}\};\qquad Q(\tau|x) = x'\theta_0(\tau);\quad  \theta_0(\tau) \in \Theta_\tau^0.
\end{equation}

\subsection{Estimator}

Motivated by the points discussed in the last two paragraphs,  we define our estimator $\h{\theta}(\tau)$ as a solution to the following optimization problem
\begin{equation}\label{eq:estimator}
 \min\limits_{\theta\in\R^{n}}\frac{1}{T_0}\sum_{t=1}^{T_0}\rho_\tau(Y_t - X_t'\theta) + \lambda\sum_{j=1}^n |\theta_j|;\qquad \tau\in\m{T},
\end{equation}
where $\lambda>0$ is a common (to all quantiles) penalty (regularizer) parameter to be specified later.  

Similar to $\theta_0(\tau)$,  $\h{\theta}(\tau)$ is not necessarily unique for each $\tau\in\m{T}$ so we let $\h{\Theta}_\tau\subset\R^n$ denote the set of all solutions to the optimization problem \eqref{eq:estimator}.  In principle, we could set $\lambda=0$ so that the second component of \eqref{eq:estimator} would vanish if we treat $n$ as fixed or growing at a slower rate as the sample size increases.  However,  it will be necessary to ensure that $\Theta_\tau^0$ and $ \h{\Theta}_\tau$ are close in the appropriate sense for cases when $n>T_0$.  In high-dimensional setup $(n>T_0)$,  the $\ell_1$-penalization component will induce a sparsity structure on $\h{\theta}(\tau)$ which will capture the assumed sparsity structure on $\theta_0(\tau)$.  We resume this discussion in Section \ref{S:Assumptions_Results} after Assumption \ref{A:CGF_smooth}.  

It is easy to verify that using the fact that $\rho_\tau(z) +  \rho_\tau(-z) = |z|$ for $z\in\R$, we can re-write \eqref{eq:estimator} as an unconstrained optimization problem in the $\lambda$-augmented variables. Precisely, \eqref{eq:estimator} is equivalent to
\begin{equation}\label{eq:estimator_unconstrained}
 \min\limits_{\theta\in\R^{n}}\sum_{t=1}^{T_0+2n}\rho_\tau(Y_t^\lambda - {X_t^\lambda}'\theta) ;\qquad \tau\in\m{T}, 
\end{equation}
where $Y_\lambda:= (Y_1^\lambda, \dots, Y_{T_0+2n}^\lambda)' = (Y_1, \dots, Y_{T_0}, 0,\dots ,0)'$ and $X_\lambda:=(X_t^\lambda, \dots, X_{T_0+2n}^\lambda)   = (X_1,\dots, X_{T_0}, \lambda T I_n ,-\lambda T I_n)$.

For each $x$, the estimated  quantile function $\tau\mapsto \widetilde{Q}(\tau|x):= x' \widehat{\theta}(\tau)$ might not be non-decreasing.  Under  Assumption \ref{A:CGF}, this is a finite sample issue but will prevent us from constructing meaningful finite sample confidence intervals whenever the quantiles ``cross''.  To circumvent  the problem, we use the  rearranged conditional quantile estimator proposed in \citet*{chernozhukov2010quantile}, which,  in our setup,  can be expressed as
\begin{equation}\label{eq:CQF_estimator}
\h{Q}(\tau|x):= \inf \left\{y\in\R: \int_0^1 \1\{x' \widehat{\theta}(\tau)\leq y\}d\tau\geq \tau \right\};\quad \tau\in\m{T}, \quad \h{\theta}(\tau) \in \h{\Theta}_\tau.
\end{equation}

Therefore, using only the pre-intervention sample $\{(Y_t,X_t):1\leq t\leq T_0\}$,  we estimate model \eqref{eq:model} using
\begin{equation}\label{eq:model_hat}
\h{\m{M}}:= \{x \mapsto \h{Q}(\tau|x):\tau\in\m{T}\}.
\end{equation}
Recall we do not observe $Y_{t}^{(0)}$ for $ T_0<t\leq T$ in our sample.  However, we can estimate all its conditional quantiles by
\begin{equation*}
 \{ \h{Q}(\tau|X_t):T_0< t\leq T,\tau\in\m{T}\}.
\end{equation*}

\subsection{Overview of the main Results}\label{SS:Overview}

Section \ref{S:Assumptions_Results} state conditions under which non-asymptotic probabilistic bounds hold for some ``distances''  between $\m{M}$ to $\h{\m{M}}$.  In this subsection, we explore the consequence of these results (refer to Section \ref{S:Results} for a formal statement) as the number of pre-intervention periods diverges ($T_0\to\infty$) while the number of post-intervention periods ($T_1:= T-T_0)$ is kept fixed. We call it a \emph{Partial Asymptotics} analysis\footnote{As opposed to the \emph{Full Asymptotics} analysis,  when the entire sample size diverges ($T\to\infty$) and the ratio $T_0/T$ is kept fixed as  $T\to\infty$}.  This approach is tailored to accommodate empirical situations where the number of pre-intervention periods $T_0$ is much larger than $T_1$ (it is not uncommon in empirical applications that we have a single or a couple of post-intervention periods), which justifies the sampling error from the estimation of $\theta_0(\tau)$ by $\h{\theta}(\tau)$ to be of smaller order when compared to the inference using the post-intervention periods.

Our first result concerns the uniform approximation of the conditional quantile function in probability.  For each post-intervention period ($T_0<t\leq T$), 
\[\h{Q}(\tau|X_t)- {Q}(\tau|X_t)=o_\P(1),\]
uniformly in $\tau\in\m{T}$, where $\m{T}$ is any compact subset of $(0,1)$. The next two results are for inference procedures. The second result  states that, for any $\alpha<0.5$,  the interval
\[\s{C}_{t}(\alpha):=[Y_t -\widehat{Q}(1-\alpha/2|X_t), Y_t -\widehat{Q}(\alpha/2|X_t)]\] 
 is an asymptotically (as $T_0\to\infty$) $1-\alpha$ confidence prediction interval for the treatment effect $\delta_t:=Y_t^{(1)}-Y_t^{(0)}$ for each post intervention period.  
 
Finally, the last result proposes a statistical test for the null hypothesis that the intervention in question had no effect, i.e., 
\begin{equation}\label{eq:null_hypothesis}
 \m{H}_0:Y_t^{(0)}=Y_t^{(1)} \text{\quad in distribution for $\;T_0<t\leq T$}
\end{equation}

To that aim, we use for test statistic,   $\phi_{\m{T},p}:= \tfrac{1}{T_1}\sum_{t=T_0+1}^T\|\1\{Y_t\leq \h{Q}(\tau|X_t)\} -\tau\|_{\m{T},p}$,  where $\|\cdot\|_{\m{T},p}$ in the $\m{L}_p$ norm with the respect to the uniform law on $\m{T}$ for any $p\in[1,\infty]$,  and show that a test that rejects when
\[\{\phi_{\m{T},p}>c_p(\alpha)\}\]
has the correct asymptotic (as $T_0\to\infty$) size $\alpha$,  where $c_p(\alpha)$ is a critical value from a pivotal distribution that only depends on $T_1$, $\m{T}$ and $p$.  Refer to Section \ref{SS:Hypothesis Testing} for further discussion on the test procedure.

The critical value can be obtained by simulation using the empirical quantile of $\phi_{\m{T},p}^0$ with $\1\{Y_t\leq \h{Q}(\tau_k|X_t)\}$ replaced by $\1\{ U\leq \tau_k\}$,  where $U$ are independent uniform (0,1) random variables.  Figure \ref{fig:norm_T1_influence} shown the estimated density for selects $\m{L}_p$ norms and interval $\m{T}\subseteq (0,1)$.

\begin{figure}[h]
\captionsetup[subfigure]{justification=centering}
\centering
\begin{subfigure}{0.45\textwidth}
    \centering
     \includegraphics[width=\textwidth]{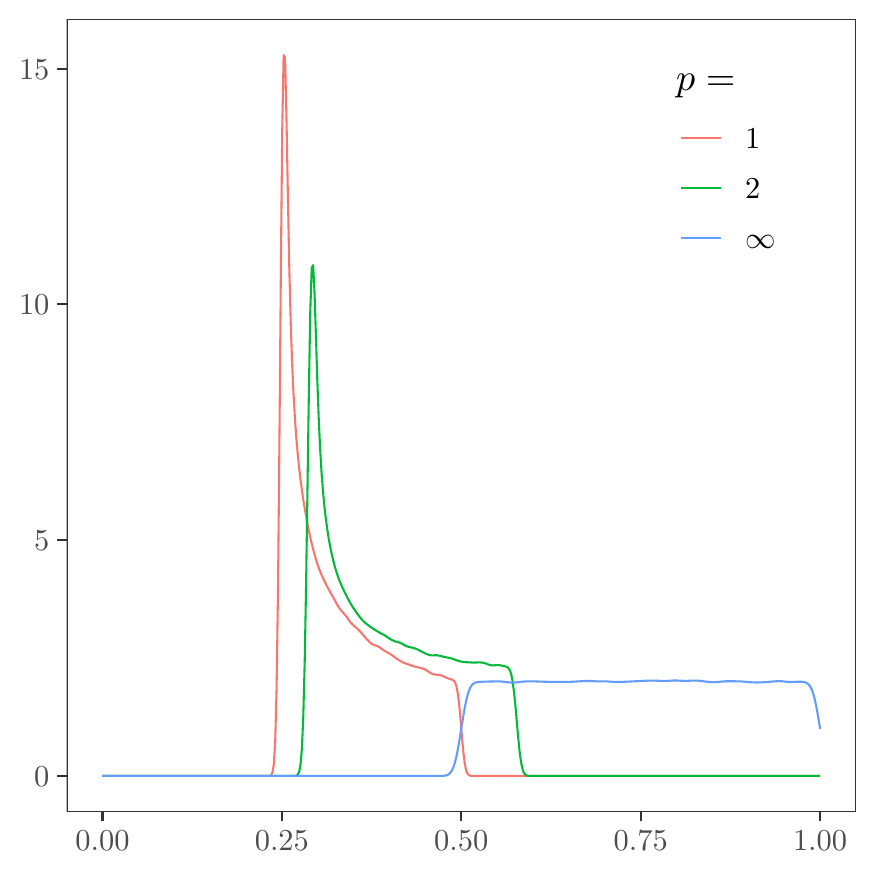}
    \caption{$\m{L}_p$-Norm Effect for selected $p$ values with $\m{T} =(0,1)$}
\end{subfigure}
\begin{subfigure}{0.45\textwidth}
    \centering
	\includegraphics[width=\textwidth]{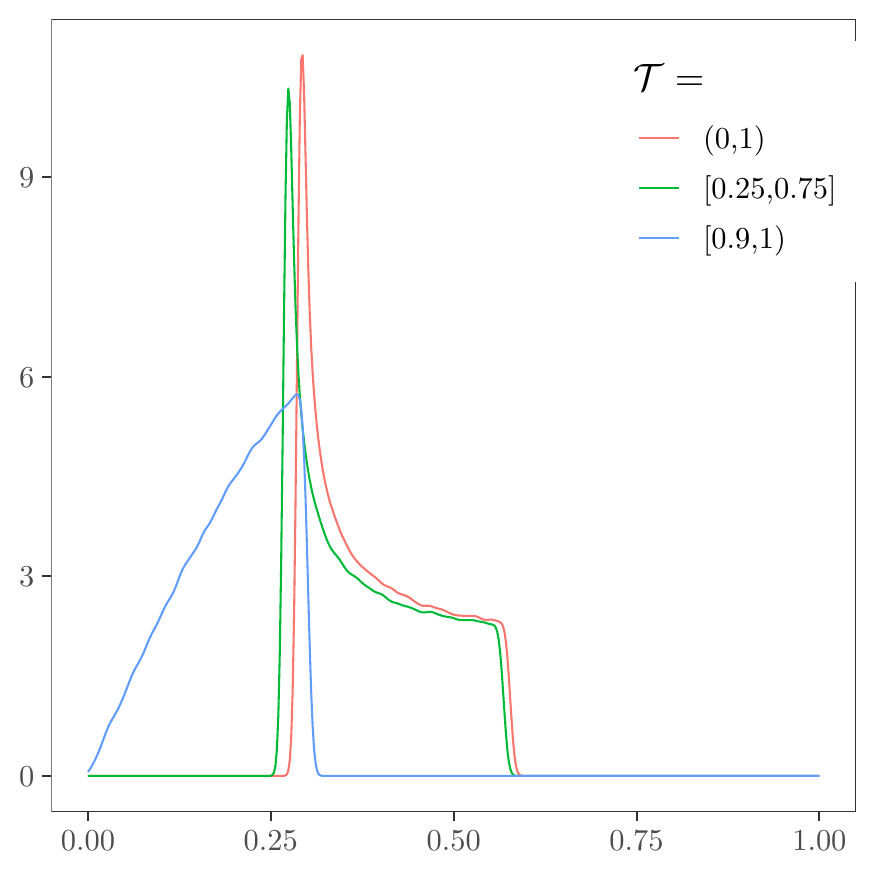}
	  \caption{Range of the interval ($\m{T}$) effect with $\m{L}_2$-Norm}
\end{subfigure}
\caption{Test Statistics density under the null hypothesis.  Simulated with 1 million draws of a uniform random variable an applying transformation  \eqref{eq:g_p}.}
\label{fig:norm_T1_influence}
\end{figure}

\section{Assumptions and Results}\label{S:Assumptions_Results}

\subsection{Technical Assumptions}

To state the results in a unified manner covering a broad range of tail distributions (polynomial and exponential), we control the tail by its Orlicz norm.  Recall that the $\psi$-Orlicz norm of a random variable $X$, which we denote by $\|X\|_\psi$,  is define as $\inf\{c>0:\E\psi(X/c)\leq 1\}$ for a convex, even, lower semi-continuous $\psi:\R\to [0,\infty]$,  such that $\psi(0)=0$. In particular, we are concerned with polynomial tails, exponential tails, and bounded variables. 

Therefore, we define $\Psi$ the class of functions  containing:
\begin{enumerate}
    \item[(i)] $x\mapsto |x|^q$ for $q\in[4,\infty)$;
    \item[(ii)] $x\mapsto  0$ if $|x|\leq L$ for $L>0$, otherwise $x\mapsto \infty$; 
    \item[(iii)] $x\mapsto \co (x\mapsto \exp(x^\mathfrak{p})-1)$ for $\mathfrak{p}\in(0,\infty)$ where $\co(f)$ denote the convex hull of $f$. 
\end{enumerate}
Cases (i) and (ii) together deal with $\m{L}_{q,\P}$-random variables whereas case $(iii)$ includes sub-Gaussian ($\mathfrak{p}=2$), sub-Exponential ($\mathfrak{p}=1$) and sub-Weibull ($\mathfrak{p}\in(0,1)$) random variables as special instances.

\begin{assumptionp}{A.4}[Tails and Variance Lower Bound]\label{A:sampling_tail} For all $1\leq i\leq n$,  $1\leq t\leq T$,
\begin{enumerate}[(a)]
\item $\|Z_{it}^{(0)}\|_{\psi}\leq B<\infty$ for some $\psi\in\Psi$ and constant $B$ possibly depending on $n$, $T$ and $\psi$;
\item $\E[ {Z_{it}^{(0)}}]^2\geq \underline{B}>0$ for some universal constant  $\underline{B}$ not depending on $T$ or $n$.
\end{enumerate}
\end{assumptionp}

Assumption \ref{A:sampling_tail}(b) bounds the second moment of $Z_{it}^{(0)}$ away from zero.  That does \emph{not} rule out multicollinearity among the units, but it is important to avoid each unit behaving like a constant as the sample size diverges. As opposed to $B$ in Assumption \ref{A:sampling_tail}(a), $\underline{B}$ is assumed to be a universal constant. While it is possible to allow $\underline{B}$ to slowly converge to zero (as a function of $n$ and $T$), that would add yet another rate term in the final oracle with no extra generality.

Recall we assume the conditional quantile function to be linear in Assumption \ref{A:CGF}.  The next assumption further imposes structure on conditional quantile function by requiring some level of sparsity (for cases when $n>T_0$) and smoothness of the map $\tau\mapsto \theta_0(\tau)$.

\begin{assumptionp}{A.5}[Conditional Quantile Function]\label{A:CGF_smooth} Assume
\begin{enumerate}[(a)]
\item $\sup_{\tau\in\m{T}}\|\theta_0(\tau)\|_0\leq s_0$ for some $1\leq s_0\leq n$
\item $\tau\mapsto \theta_0(\tau)$ is Lipschitz with constant no greater than $(T\lor n)^c$ for some constant $c>0$.
\end{enumerate}
\end{assumptionp}

Specifically,  the assumption above requires that at most $s_0
\leq n$ elements of $\theta_0(\tau)$ differ from zero in the conditional quantile function. Notice that these non-zeros elements might differ from quantile to quantile.  The assumption is similar to the one considered in \citet*{HeWangHong2013} in the context of nonlinear quantile regression models.  In a factor model structure with a sparse loading matrix, one expects the non-zeros elements to be the same across all quantiles. We might have a dense model in that $s_0=n$ provided that $n$ grows slower than $T$ (low-dimensional setup). Otherwise, $s_0$ must be strictly smaller than $n$. The smoothness requirement is not strict as we allow the Lipschitz constant to grow polynomial with $T\lor n$. For instance,  it is implied whenever $\sup_j|\theta_j(\tau)-\theta_j(\tau')|\leq C|\tau-\tau'|$ for  all $\tau,\tau'\in\m{T}$.
 
\begin{assumptionp}{A.6}[Conditional Density]\label{A:conditional_density}
For every $x$ in the support of $X_t$ for $1\leq t\leq T$
\begin{enumerate}[(a)]
\item $y\mapsto F(y|x)$ admits a continuously differentiable density,  which we denote by $f(y|x)$, and $f(y|x)$ and its first partial derivative $\dot{f}(y|x)$ with respect to $y$ are upper bounded by some constant $\bar{f}<\infty$.
\item $\inf_{\tau\in\m{T}}f(Q(\tau|x)|x)\geq \underline{f}>0$ some constant $ \underline{f}$.
\end{enumerate}
\end{assumptionp}

Assumption \ref{A:conditional_density} above is a regularity condition on the conditional distribution of $Y_t|X_t$. Parts (a) and (b) are used to ensure the (restricted) strong convexity of the population objective function in the neighborhood of the quantiles of interest. Similar conditions are needed to ensure the parameter identification even in the fixed $n$ case.
 
We need to establish an additional piece of notation to state our next assumption clearly and concisely.  For $x\in\R^n$ and index set  $A\subseteq \{1,\dots,n \}$ we write $x_A$ for the vector $x$ with the entries index by $A^c:=\{1,\dots,n \}\setminus A$ set to zero.  For each $\tau\in\m{T}$,  let $\m{S}_\tau\subseteq \{1,\dots,n \}$ denote the support of $\theta_0(\tau)$, i.e., $\m{S}_\tau = \{j:\theta_{0,j}(\tau)\neq 0\}$, and define the family of  cones $\m{C}:=\{\m{C}_\tau:\tau\in\m{T}\}$ in $\R^n$ as follows
\[
 \m{C}_\tau:=\{x\in\R^n:\| x_{\m{S}^c}\|_1\leq 3\| x_{\m{S}}\|_1\}.
\]
 
 \begin{assumptionp}{A.7}[Compatibility Condition]\label{A:comp_cond}
There is a positive constant $\varphi$ possibly depending on $n$ and $T$ such that 
\[
\inf\limits_{\tau\in\m{T}}\inf\limits_ {x\in \m{C}_{\tau},x\neq 0} \frac{x'\Sigma x |\m{S}_\tau|}{\|x_{\m{S}_\tau}\|_1^2}\geq \varphi^2,
\]
where $\Sigma:=\E(\tfrac{1}{T_0}\sum_{t=1}^{T_0} X_tX_t')$.
 \end{assumptionp}
 
The compatibility condition is standard in LASSO literature; refer to Chapter 6 of \citet*{buhlmann2011statistics} for further details and applications. It is also closely related to the well-known restricted eigenvalue condition on the $\Sigma$. As a matter of fact, the latter is implied by the former and is sufficient for a restricted identification of $\theta_0(\tau)$.  Furthermore,  Assumption \ref{A:conditional_density} combined with Assumption \ref{A:comp_cond}(a) ensure that $\theta_0(\tau)$ as a maximizer of $\theta\mapsto J(\tau):=\E(\rho_\tau(Y_t-Z_t'\theta)$ is unique restricted to the cone $\m{C}_\tau$, uniformly in $\tau\in\m{T}$.  Clearly, a sufficient condition for Assumption \ref{A:comp_cond} would be the minimum eigenvalue of $\Sigma$ to be bounded away for zero by an absolute constant.

\subsection{Main Results}\label{S:Results}

In this subsection, we present the two main results of the paper, including two useful corollaries. Recall the serial dependence measure in terms of $\beta$-mixing coefficients defined in \eqref{E:beta_mixing_coef}. With a slight abuse of notation, define its continuous version by the function $\beta:\R^+\to[1/4,0]$ given by $\beta(t):= \beta_{\lfloor t\rfloor}$ for $t\geq 0$. Since $\beta(t)$ is a non-increasing function, it admits a left-inverse, which we denote by $\beta^{-1}$. We use the function $\beta^{-1}$ to present the rate dependency of the mixing coefficients of the series $\{Z_t^{(0)}\}$. Similarly the tails of $\{Z_t^{(0)}\}$ are controlled via the $\psi$-Orlicz norm for $\psi\in\Psi$ (refer to Assumption \ref{A:sampling_tail}) and $\psi^{-1}$ denotes the (left) inverse of $\psi$ restricted to $\R_+$.

\begin{theorem}\label{thm:main} Under Assumptions \ref{A:DGP}-\ref{A:comp_cond},   suppose further that
\[\lambda = C\vartheta^2 B\underline{f}\varphi^2 r, \]
 where $C>0$ is a universal constant and
 \[
 r := r(n,T_0,\beta,\psi):=\psi^{-1}(nT_0)\sqrt{\frac{(\beta^{-1}([\vartheta T_0)]^{-1})\lor 1)\log n}{T_0}};
 \]
and $ [B\psi^{-1}(nT_0)]^2s_0r\leq \frac{\underline{f}^2\varphi^2}{ \overline{f}}$.  Then, for $\vartheta\geq 1$,
\begin{equation}\label{eq:main_thm_res_1}
    \P\left(\sup_{\tau\in\m{T}}\|\widehat{\theta}(\tau) - \theta_0(\tau)\|_1\gs\lambda s_0\right)  \ls n^{1-\vartheta^2} + \frac{1}{\vartheta}
\end{equation}
and, for any $b>0$,
\begin{equation}\label{eq:main_thm_res_2}
    \P\left(\sup_{\tau\in\m{T}}| \widehat{Q}(\tau|X_t)-Q(\tau|X_t)|\gs bB \lambda s_0 \psi^{-1}(n) \right)\ls n^{1-\vartheta^2} + \frac{1}{\vartheta}+ \frac{1}{b}.
\end{equation}
Also,
\begin{equation}\label{eq:main_thm_res_3}
    \sup_{\tau\in\m{T}}|\P(Y_t^{(0)}\leq \widehat{Q}(\tau|X_t)) -\tau| \ls\inf_{\vartheta \geq 1, b>0} \left[  n^{1-\vartheta^2} + \frac{1}{\vartheta}+ \frac{1}{b} + b B \lambda s_0 \psi^{-1}(n)\right].
\end{equation}
\end{theorem}

The first result in Theorem \ref{thm:main} is related to \citet*{belloni2011} and \citet*{zheng2015}  with essential contrasts. First, note that \eqref{eq:main_thm_res_1} is an $\ell_1$-convergence rate as opposed to $\ell_2$ in the aforementioned papers. This difference is vital to be able to derive result \eqref{eq:main_thm_res_2}, which does not follow from the second inequality in Theorem 2 in \citet*{belloni2011}. Second, Theorem 1 accommodates non-independent data via the $\beta$-mixing coefficients with an arbitrary decay and a wide range of distribution tails. Notice that for geometric $\beta$-mixing decay (for instance, most ARMA processes), we pay the price of an extra $\log T_0$ term in the numerator as compared to the independent counterpart. Also, we recover the $\ell_1$-sharp rate $s_0\sqrt{\log n/T_0}$ uniformly in $\tau\in\m{T}$ under independence ($\beta_m =0$ for $m\geq 1$) and absolutely bounded random variables (c.f. Corollary 6.2 in \citet*{buhlmann2011statistics}).

Furthermore, note that our estimator differs from \citet*{belloni2011}, who, in the independence setting, consider a self-normalized process by using a penalty specific for each variable equal to its sample standard deviation. From the practical point of view, this can always be achieved by normalizing the variable before running the LASSO estimator. However, this modification considerably complicates the theoretical analysis of the estimator. That is because, with a self-normalized process under dependency, we need to guarantee the Berbee's Coupling construction (refer to Proposition 2 in \citet*{DMR_1995}) conditional on the $X_{t}$. This explains why in \citet*{belloni2011}, the tail $X_t$ influence appears on the condition that the sample normalization converges uniformly to the population normalization. Refer to the definition of $\gamma$ in Condition D.3 in \citet*{belloni2011}. In most cases, there will be a dependence on the number of covariates implicit in $\gamma$. Our estimator also differs from \citet*{zheng2015}, who consider a $\ell_1$-adaptive penalization in ultra high-dimensional setup. Under independence and compact support assumptions, the authors improve the  $\ell_2$-rate appearing in \citet*{belloni2011}  from $\sqrt{s_0\log n/T_0}$ to $\sqrt{s_0\log s_0/T_0}$ which allow them to accommodate ultra-high dimensional, i.e., when $\log n$ increases polynomial in $T$.

For our purposes, \eqref{eq:main_thm_res_1} is an intermediate step for the final results appearing in Theorem \ref{thm:main}, which gives non-asymptotic bounds for the estimation of the conditional quantile function of the counterfactual uniformly over the quantiles in $\m{T}$. The results are presented both in terms of its ``closeness'' to the true conditional quantile function, inequality \eqref{eq:main_thm_res_2}; and the probability of deviation from it when we use the estimated quantile function instead of the true one, inequality  \eqref{eq:main_thm_res_3}. The term $n^{1-\vartheta^2}$ captures the variation in the ``empirical process'' component, which is controlled by $\lambda$, whereas the terms of the form $\psi^{-1}(\cdot)$ represents the dependency on the tail of $X_t$. For instance,  if $X_t$ has $q\geq 4$ finite moments then $\psi^{-1}(x) = x^{1/q}$, whereas in  the sub-Gaussissian case, we may take $\psi^{-1}(x)$ to be of order $\sqrt{\log (x)}$.

The main consequence of Theorem \ref{thm:main} is that we can estimate and draw inference on the intervention effect. For instance, from \eqref{eq:main_thm_res_2} we conclude under mild conditions that $\widehat{Q}(\tau|X_t)-Q(\tau|X_t) = O_\P\left(s_0\psi^{-1}(n)r\right)$ uniformly in $\tau\in\m{T}$.  In that case,  $\widehat{Q}(\tau|X_t)$ is a consistent  estimator of the $\tau$ conditional quantile of the the counterfactual uniformly in $\tau\in\m{T}$,  provided that $s\psi^{-1}(n)r=o(1)$ as $T_0\to\infty$.  We could for instance  take the median $\{\widehat{Q}(0.5|X_t)\}_{t> T_0}$ as the point estimate for the counterfactual $\{Y_t^{(0)}\}_{t>T_0}$. Furthermore, since we have a model for each conditional quantile, we can construct confidence bands of asymptotically correct size for the unobservable $Y_t^{(0)}$ after the intervention.  As a concrete example, take the sub-Gaussian case ($\psi(x) = \exp(x^2)-1$) with geometric mixing $\beta_m\ls \exp(-cm)$, the last result becomes
\[
\sup_{\tau\in\m{T}}|\P(Y_t^{(0)}\leq \widehat{Q}(\tau|X_t)) -\tau| \ = O\left(s_0\sqrt{\frac{\log (n\lor T_0)\log T_0)(\log n)^2}{T_0}} \right).
\]

We formalize the discussion in the last paragraph in the following Corollary.

\begin{corollary}\label{C:CI} Under the same conditions of Theorem \ref{thm:main},  if further  $B$, $\overline{f}$, $\underline{f}$ and $\varphi$ are fixed constant (not depending on $n$ and $T$) 
 then for each $T_0<t\leq T$:
\[
\sup_{\tau\in\m{T}}| \widehat{Q}(\tau|X_t)-Q(\tau|X_t)|=O_\P(s_0\psi^{-1}(n)r).
\]
Additionally, under  Assumption \ref{A:identification} and $s_0\psi^{-1}(n)r = o_\P(1)$,  we have that
\[\mathscr{C}_{t}(\alpha):=[Y_t -\widehat{Q}(1-\alpha/2|X_t), Y_t -\widehat{Q}(\alpha/2|X_t)]\]  is an asymptotically $1-\alpha$ confidence interval for the treatment effect $\delta_t:=Y_t^{(1)}-Y_t^{(0)}$  for each post intervention period;  in the sense that, for all $T_0<t\leq T$,  we have
\[\sup_{0\leq \alpha<0.5}\left|\P\left( \delta_t \in\mathscr{C}_{t}(\alpha)\right)-(1-\alpha)\right|=o(1) \quad\text{as}\quad T_0\to\infty.\]
\end{corollary}

\subsubsection{Hypothesis Testing}\label{SS:Hypothesis Testing}

For now, suppose we have only a single post-intervention period $(T_1=1)$ and $Q(\cdot|\cdot)$ is known.  Consider a test that rejects when $Y_T< Q(\tau_1|X_t) $ or  $Y_T> Q(\tau_2|X_t) $ for some $\tau_1\leq \tau_2\in\m{T}$.  The test has the correct size $\alpha$ as soon as we choose $\tau_1,\tau_2$ such that $\tau_2-\tau_1\geq 1 - \alpha$.   Clearly,  there are infinite as many tests indexed by the pairs $(\tau_1,\tau_2)\in\m{T}^2$.  If we take the rejection region to be symmetric around zero,  we can show the test is powerful against alternatives of the type $Y_T^{(1)} = Y_T^{(0)} + c$ for some constant $c\in\R$.  Consider the case when the treatment is mean preserving but reduces the dispersion; for instance, $Y^{(1)}_T = \sigma Y^{(0)}_T$ for some $\sigma>0$,  the test now looses power for $0< \sigma<1$ and it can its power approach size as $\sigma$ gets smaller.  

In this situation, it would be better to consider the rejection region centered instead of on the tails. However, we loss power against alternatives when $\sigma>1$ or when the treatment is of the type $Y_T^{(1)} = Y_T^{(0)} + \epsilon$ for some non-degenerate random variable $\epsilon$. We could also consider different rejection regions, such as non-symmetric,  non-connected, or any other shape, aiming to gain power against more complex alternatives.  There are several (parametric) weak null hypotheses implied by \eqref{eq:null_hypothesis}, notably $\E Y_T^{(1)}=\E Y_T^{(0)}$, the classical null hypothesis on ATT estimation in cross-sectional setup. 

We propose to test the null \eqref{eq:null_hypothesis} by testing the stability of all the quantiles of the conditional distribution of $Y_T$ given $X_T$ before and after the intervention.  To expose the idea, consider the stochastic process $\{M_T^0(\tau):\tau\in \m{T}\}$ where $M^0_T(\tau):=\1\{Y_T\leq  Q(\tau|X_T)\}-\tau$ and $\m{T}$ is any compact subset of $(0,1)$.  By construction,  under the null hypothesis,  $\{M_T^0(\tau):\tau\in \m{T}\}$ is  a zero-mean process and $\1\{Y_T\leq  Q(\tau|X_T)\}$ is distributed as a $\text{Bernoulli} (\tau)$ for each $\tau\in\m{T}$. Thus,  $\{M_T^0(\tau):\tau\in \m{T}\}$ can be seen as a process of dependent over $\tau\in\m{T}$ centered Bernoullis.

Obviously,  any hypothesis test based on $\{M_T^0(\tau):\tau\in \m{T}\}$ is infeasible. However,  apart from the unobservable mapping $\tau\mapsto Q(\tau|\cdot)$,  which can be estimated  (uniformly) by $\widehat{Q}$,  the distribution of $\{M_T^0(\tau):\tau\in \m{T}\}$ is pivotal under the null.  To see that, notice that  $\{Y_T\leq Q(\tau|X_T)\} = \{U\leq \tau\}$ almost surely under the null, where $U$ denotes a uniformly distributed random variable over the interval $(0,1)$. Then, $\{M_T^0(\tau):\tau\in \m{T}\}$ can be (almost surely) equivalently expressed as $\{\1\{U\leq \tau\}-\tau:\tau\in\m{T}\}$. As a consequence,  finite-sample (in terms of post-intervention periods) inference is possible as long as we can compute critical values for $\|M^0_{T}\|_{p,\m{T}}$,  where $\|\cdot\|_{p,\m{T}}$ denotes the $\m{L}_p$ norm for $p\in[1,\infty]$ under the uniform measure over $\m{T}$. 

For a given $\m{T}$ and $p\in[1,\infty]$, the proposed test-statistic takes the form
\begin{equation}\label{eq:test_statistic}
\phi_{{\m{T}},p}:=\tfrac{1}{T_1}\sum_{t>T_0}^T\|M_t\|_{p,\m{T}}; \quad  M_t(\tau):= \1\{Y_t\leq  \widehat{Q}(\tau|X_t)\}-\tau,\quad \tau\in\m{T}. 
\end{equation}

The following result bounds the difference in the distribution of $\phi_{{\m{T}},p}$ and $\phi_{{\m{T}},p}^0$ where the latter is the former with $\widehat{Q}$ replaced by $Q$. The result is a ``convergence in distribution type'' except that the target distribution also depends on the sample size $T$ and number of peers $n$, and, for that reason, it does not necessarily weak-converge in the usual sense.

\begin{theorem}\label{T:test} Let $\m{T}\subset (0,1)$ be a compact and $p\in[1,\infty]$.  Under the same conditions of Theorem \ref{thm:main} and $\m{H}_0$, there is a constant $C_p$ only depending on $p$ such that
\begin{align*}
|\P(\phi_{{\m{T}},p} \leq x) - \P(\phi_{{\m{T}},p} ^0\leq x)|  &\leq \\C_p\inf_{\vartheta\geq 1, b>0} &\left[bB\lambda s_0\psi^{-1}(n) +  T_1\left(n^{1-\vartheta^2} + \frac{1}{\vartheta}+ \frac{1}{b}\right)\right],
\end{align*}
for every continuity point of $x\mapsto \P(\phi_{{\m{T}},p}^0 \leq x)$; where $T_1:=T-T_0$ is the number of post-intervention periods.
\end{theorem}

We propose a test that rejects $\m{H}_0$ whenever  $\phi_{{\m{T}},p} > c$ for some constant $c$ and given $\m{T}$ and $p\in[1,\infty]$. To control for the test size $\alpha$, say,  we set $c=c_{\m{T},p}(\alpha/T_1)$ where $c_{\m{T},p}(\mathfrak{q})$ denotes the $\mathfrak{q}$-quantile of $\|M_t^0\|_{p,\m{T}}$ under the null. Notice that that the distribution of $M_t^0$ is independent of $t$ under Assumption \ref{A:CGF}.

\begin{corollary} For $\alpha>0$, $p\in[1,\infty]$ and some (large enough) compact interval $\m{T}\subset(0,1)$ we have, under the same conditions of Theorem \ref{T:test}, 
\[
\P\big[\phi_{{\m{T}},p}>c_{\m{T},p}(\alpha/T_1)\big] - \alpha  \ls\inf_{\vartheta\geq 1, b>0} \left[bB\lambda s_0\psi^{-1}(n) +  T_1\left(n^{1-\vartheta^2} + \frac{1}{\vartheta}+ \frac{1}{b}\right)\right].
\]
In particular, the test has the correct asymptotic size $\alpha$ provided that the rate condition in Corollary \ref{C:CI} is satisfied. 
\end{corollary}

The condition that interval $\m{T}$ must be large enough is to ensure that it includes the  $1-\alpha$ quantile of $\phi_{\m{T},p}^0$, which then becomes a continuity point and the result above follows directly from Theorem \ref{T:test} and the union bound.

To implement the hypothesis test described above, we need to compute the critical value $c_{\m{T},p}(\cdot)$ or p-values. Fortunately,  as previously mentioned,  the distribution of $\{M^0_t(\tau):\tau\in\m{T}\}$ is pivotal under the null so that a direct Monte Carlo simulation can be used.  For every $u\in[0,1]$ define
\begin{equation*}
g_{\m{T},p}(u) :=\|\1\{u\leq\tau\}-\tau\|_{\m{T},p}:=\begin{cases}
\left(\int_{\m{T}}|\1\{u\leq\tau\}-\tau|^p d\tau \right)^{1/p}&;  \quad p\in[1,\infty)\\
\sup_{\tau\in\m{T}} |\1\{u\leq\tau\}-\tau|\quad &; \quad p=\infty.
\end{cases}
\end{equation*}

Therefore, $\|M_t^0\|_{\m{T},p}$ and $g_{\m{T},p}(U)$ share the same distribution under the null for every $t>T_0$ where $U$ is a uniform random variable on the unit interval. If we take $\m{T}$ to be a closed interval with endpoints $0<\underline{\tau}\leq \overline{\tau}< 1$, since $\tau\mapsto \1\{u\leq\tau\}-\tau$ is piece-wise linear for each $u\in[0,1]$, we can solve the integral analytically for $p\in[1,\infty)$ and the suprema appearing in the case $p=\infty$ is straightforward to compute. Specifically, the last display becomes
\begin{equation}\label{eq:g_p}
g_{\m{T},p}(u)=\begin{cases}
\left(\tfrac{1}{p+1}\left[{\tau^*(u)}^{p+1} - \underline{\tau}^{p+1} - ( 1 - \overline{\tau})^{p+1} + (1-\tau^*(u))^{p+1} \right]\right)^{1/p} &;  \quad p\in[1,\infty)\\
 \tau^*(u)\lor (1-\tau^*(u)) \quad &; \quad p=\infty,
\end{cases}
\end{equation}
where $\tau^*(u):=\underline{\tau}\lor(u\land\overline{\tau})$.

Equation \eqref{eq:g_p} is key to computing critical values by simply drawing a large number of uniform random variables in a Monte Carlo simulation.  Figure \ref{fig:norm_T1_influence} above was constructed using this approach.  Furthermore,  since $g_{\m{T},p}$ is continuously differentiable except at finitely many points, we conclude that $\phi_{\m{T},p}^0$ is an absolutely continuous random variable for each $p\in[1,\infty]$.

Also, in practice, we cannot compute the test statistics $\phi_{\m{T},p}$ defined by \eqref{eq:test_statistic} for an infinite set $\m{T}$.  So, we take a large (but finite) $\m{T}$ as an approximation.  For instance,  $\m{T}$ could be the set of the $n_\m{T} \geq 2$ equally-spaced quantiles from the lowest chosen quantile,  $0<{\underline{\tau}}$,  to the largest chosen quantile ${\overline{\tau}}<1$.  For finite  $\m{T}$ we have
\[
\phi_{\m{T},p} = \begin{cases}
T_1^{-1}\sum_{t>T_0}^T\left(|\m{T}|^{-1}\sum_{k=1}^{|\m{T}|}\left|\1\{Y_t\leq \h{Q}(\tau_k|X_t)\} -\tau_k\right|^p\right)^{1/p}; & p\in [1,\infty)\\
T_1 ^{-1}\sum_{t>T_0}^T \max_{1\leq k\leq |\m{T}|}\left|\1\{Y_t\leq \h{Q}(\tau_k|X_t)\} -\tau_k\right|; & p=\infty.\\
\end{cases}
\]

Finally, the choice of the norm $p\in[1,\infty]$ ultimately depends on the expected alternative. Similarly to the non-parametric hypothesis test situation, it is not possible to have the optimum $p$ in the sense of maximizing the test power against all potential alternatives. Recall the alternative in our case is defined by $Y_t^{(0)} \neq Y_t^{(1)}$ in distribution for the post-intervention periods. While it is possible to compute the test power function for a given alternative, i.e., for a fixed pair of laws $(Y_t^{(0)}, Y_t^{(1)})$, in most empirical applications, the \emph{expected} distribution effect is usually unknown by the econometrician.

Heuristically, we have that when all (or most) of the quantiles of $Y_t$ are expected to change due to the intervention, for instance, when a shift to distribution location occurs, smaller norms such as $p\in\{1,2\}$ are more powerful when compared to $p=\infty$. The opposite is true, i.e., when the change is localized around a single quantile, for instance, only the extreme right tail of the distribution changes, then $p=\infty$ is preferable as opposed to $p\in\{1,2\}$. Both behaviors are intuitive and are corroborated in the simulations. As a rule of thumb, we recommend computing the hypothesis test for $p\in\{1,2,\infty\}$ to guide one's conclusion in empirical investigation. We follow this recommendation in the empirical illustration presented in Section \ref{S:application}.

\section{Monte Carlo Simulation}\label{S:montecarlo}

 We conducted a Monte Carlo study to evaluate the performance of the estimator \eqref{eq:estimator} and the test \eqref{eq:test_statistic} in finite-sample.  The simulated DGP is a linear location and scale model given by
\begin{equation}\label{eq:MC_DGP}
  Y_t =X_t' \beta_0 + \sigma(X_t)U_t,
\end{equation}
where  $X_t=(1,X_{2t},\dots, X_{nt})'$ is a  $n$-dimensional random vector with covariance structure $\Omega:=(\omega_{ij})$ among the last $n-1$ entries given by $\omega_{ij} = \rho_X^{|i-j|}$ for $\rho_X\in[0,1)$ and $U_t$ is zero-mean unit-variance absolutely continuous random variable independent of $X_t$ with cdf $G$ . Also, $\{(X_t,U_t):t\geq 1\}$ is an autoregressive process with coefficient $\rho_Z\in[0,1)$ for each element. Finally, $\beta_0=(\alpha_0, \t{\beta}'_0)'\in\R^{n}$ is a sparse vector  and  $x\mapsto \sigma(x)> 0$ the scale function both to be specified later.
 
The  conditional quantile function of \eqref{eq:MC_DGP} becomes
 \[Q(\tau|x) = x' \beta_0  + \sigma(x)G^{-1}(\tau).\]
Clearly, we have a correctly specified linear quantile model (Assumption \ref{A:CGF}) if the scale function is a constant $\sigma(x)=\sigma^*> 0$ or, more generally when the scale function is linear on $x$, i.e.,  $\sigma(x)=\sigma^* + x'\gamma$ for some  $\gamma\in\R^{n}$ provided that $\gamma'x>0$ for all $x$ in the support of $X_t$. In that case,  the the conditional quantile becomes $Q(\tau|x) = x' \theta_0(\tau)$ where $\theta_0(\tau) = \big( \alpha_{0} + \sigma^*G^{-1}(\gamma), (\t{\beta}_0 + G^{-1}(\tau)\gamma)'\big)'$  We simulate \eqref{eq:MC_DGP} with non-linear scale function to evaluate departures from Assumption \ref{A:CGF}.  The conditional density of $Y_t$ given $X_t=x$ is given by $f(y|x) = g(y)/\sigma(x)$ where $g$ denotes the density of $G$.  Assumption \ref{A:conditional_density} is fulfilled provided that $\sigma(x)$ is bounded away from zero and infinity over the support of $X_t$, and $g$ is bounded away from zero and infinity over the support of the conditional quantiles for a given compact interval $\m{T}\subset (0,1)$.

 The baseline for our simulation is $n = 200$ units ( including the treated one) with $s_0=5$ relevant units for all quantiles, $T_0=100$ pre-intervention periods,  $T_1 = 3$ post-intervention periods.  All marginal distributions are Gaussian, and the correlation coefficient among the regressors is $\rho_X=0.25$ and the serial correlation coefficient $\rho_Z=0.5$.  The linear quantile model is correctly specified.  For comparison, we show in the first column of Tables \ref{thm:main} and \ref{T:test} the results for the ``Oracle'' estimation, i.e.,  the estimation using unpenalized quantile regression if we knew which are the relevant regressors ($\m{S}_\tau$ is known for each $\tau\in\m{T}$).
 
 We also consider different configurations around a \emph{baseline} scenario. We investigate the influence of decreasing the number of pre-intervention periods to $T_0 = 80$ and $50$. An increase in the number of peers to $n=300$ and $500$.  Change the Signal-to-Noise ratio (SNR) to $\text{SNR}=0.1$ and $10$.  Alter the distribution of the $G$  to a Chi-squared and t-distribution with 4 and 10 degrees of freedom, respectively.  Finally, we consider a missecifed case where $\sigma(x) = 1 + (\gamma'x)^2$ with $\gamma= (1,1/2,\dots, 1/5)$.  See the header of Table \ref{tab:MC_estimation} and  \ref{tab:MC_HT} for a summary.  For all cases,  $\m{T}$ is set to be an equally-spaced grid over $(0,1)$ containing 100 points.  

\begin{table}[h]
\footnotesize
\centering
\begin{threeparttable}
\caption{Monte Carlo Simulation: CQF Estimation Error \label{tab:MC_estimation}}
\begin{tabular*}{\textwidth}{@{\extracolsep{\fill} } c c c c c c c c c c c c  c }
\toprule
& & &\multicolumn{2}{c}{$T_0=$} & \multicolumn{2}{c}{ $n=$} & \multicolumn{2}{c}{SNR}&\multicolumn{2}{c}{$G$} \\
  \cmidrule{2-3}\cmidrule{4-5} \cmidrule{6-7} \cmidrule{8-9}\cmidrule{10-11}
 Nominal &Oracle & Baseline & $80$ & $50$ & $300$ & $500$ & $0.1$ & $10$ & $\chi^2_4$ & $t_{10}$ & Misspec.\\
   \midrule
  \multicolumn{10}{l}{\textit{Panel A: CQF Estimation Error}} \\
 \midrule
   $\m{L}_p$ \\
   \midrule
$p=1$ & 1.23 & 1.38 & 1.33 & 1.31 & 1.40 & 1.36 & 2.05 & 1.81 & 1.42 & 1.40 & 4.29 \\ 
  $p=2$ & 1.37 & 1.66 & 1.65 & 1.62 & 1.68 & 1.65 & 2.55 & 2.32 & 1.69 & 1.72 & 7.50\\
  $p=\infty$ & 2.91 & 3.41 & 3.20 & 3.05 & 3.36 & 3.30 & 4.89 & 5.60 & 3.63 & 3.54 & 22.48 \\ 
   \midrule
  \multicolumn{10}{l}{\textit{Panel B: Empirical (Conditional)} Quantiles} \\
 \midrule
   $\tau$\\
   \midrule
  0.1 & .135 & .112 & .113 & .108 & .108 & .107 & .106 & .108 & .105 & .105 & .113 \\ 
  0.2 & .231 & .206 & .209 & .206 & .214 & .202 & .198 & .215 & .210 & .201 & .208 \\ 
  0.3 & .315 & .294 & .312 & .302 & .317 & .299 & .293 & .317 & .302 & .301 & .304 \\ 
  0.4 & .412 & .403 & .410 & .400 & .420 & .390 & .388 & .415 & .411 & .394 & .405 \\ 
  0.5 & .500 & .498 & .510 & .494 & .512 & .496 & .497 & .511 & .505 & .494 & .511 \\ 
  0.6 & .596 & .599 & .595 & .588 & .610 & .594 & .588 & .605 & .595 & .588 & .607 \\ 
  0.7 & .685 & .690 & .693 & .692 & .707 & .693 & .693 & .701 & .701 & .686 & .699 \\ 
  0.8 & .776 & .783 & .799 & .789 & .804 & .792 & .787 & .797 & .793 & .785 & .795 \\ 
  0.9 & .865 & .879 & .894 & .880 & .901 & .893 & .888 & .890 & .899 & .887 & .885 \\ 
  \midrule
 \multicolumn{10}{l}{\textit{Panel C: Confidence Interval Empirical Coverage}} \\
 \midrule
   Nominal\\
   \midrule
0.9 & .829 & .876 & .890 & .888 & .895 & .881 & .881 & .884 & .886 & .890 & .875 \\ 
  0.95 & .886 & .938 & .926 & .922 & .946 & .937 & .938 & .941 & .943 & .942 & .939 \\ 
  0.99 & .913 & .980 & .975 & .960 & .979 & .976 & .981 & .979 & .979 & .984 & .974 \\ 
\bottomrule
\end{tabular*}
\begin{tablenotes}
\footnotesize
\item \underline{Note}: Results based on 1,000 replications for each case.  Baseline scenario ($n=200,T=100, T_1 = 3, s_0=5,  G=N(0,1), \rho_X =0.25, \text{SNR}=1, \beta_{0j}=1/j$ for $1\leq j\leq s_0$). $\m{T}$ is  equally-spaced grid over $(0,1)$ containing 100 points. Penalty parameter selected as per \citet*{belloni2011} with confidence level 90\%. Signal-to-Noise ratio (SNR); Misspecified model (Misspec)  has $\sigma(x) = 1 + (\gamma'x)^2$ with $\gamma= (1,1/2,\dots, 1/5)$.
\end{tablenotes}
\end{threeparttable}
\end{table}

In every replication,  we compute the estimator defined in  \eqref{eq:estimator} and the test statistics described in \eqref{eq:test_statistic} over $\m{T}$.  Table \ref{tab:MC_estimation} organizes the CQF estimation results.  Panel A evaluates the estimator performance in terms of the distance between  $\tau\mapsto \widehat{Q}(\tau|X_t)$ and $\tau\mapsto Q(\tau|X_t)$ in the $\m{L}_p$ norm for $p\in\{1,2,\infty\}$ over the uniform distribution in the grid. The result is then averaged over the post-intervention periods. Panel B presents the empirical conditional quantiles,  i.e., the conditional quantiles implied by the estimated CQF.  Finally,  Panel (C) shows the empirical coverage of confidence intervals created as described in Corollary \ref{C:CI}.  Table \ref{tab:MC_HT} collects the results of the hypothesis test based on test statistic \eqref{eq:test_statistic} for the norm $p\in\{1,2,\infty\}$. It shows the rejection rates under the null (empirical size), computed via the simulated null distribution using $10,000$ replications of \eqref{eq:g_p}.

\begin{table}[h]
\footnotesize
\centering
\begin{threeparttable}
\caption{Monte Carlo Simulation: Empirical Size} \label{tab:MC_HT}
\begin{tabular*}{\textwidth}{@{\extracolsep{\fill} } c c c c c c c c c c c c  c }
\toprule
& & &\multicolumn{2}{c}{$T_0=$} & \multicolumn{2}{c}{ $n=$} & \multicolumn{2}{c}{SNR}&\multicolumn{2}{c}{$G$} \\
  \cmidrule{2-3}\cmidrule{4-5} \cmidrule{6-7} \cmidrule{8-9}\cmidrule{10-11}
 Nominal &Oracle & Baseline & $80$ & $50$ & $300$ & $500$ & $0.1$ & $10$ & $\chi^2_4$ & $t_{10}$ & Misspec.\\
  \midrule
 \multicolumn{10}{l}{$\m{L}_1$-norm} \\
 \midrule
0.1 & .157 & .120 & .119 & .125 & .101 & .101 & .135 & .118 & .106 & .113 & .126 \\ 
  0.05 & .095 & .073 & .055 & .069 & .052 & .054 & .066 & .067 & .058 & .053 & .062 \\ 
  0.01 & .036 & .017 & .014 & .022 & .015 & .015 & .019 & .014 & .009 & .017 & .011 \\ 
   \midrule
 \multicolumn{10}{l}{$\m{L}_2$-norm} \\
 \midrule
 0.1 & .150 & .110 & .116 & .130 & .096 & .114 & .131 & .129 & .115 & .107 & .122 \\ 
  0.05 & .083 & .064 & .061 & .072 & .052 & .059 & .063 & .061 & .055 & .055 & .058 \\ 
  0.01 & .018 & .020 & .020 & .0.19 & .018 & .015 & .018 & .014 & .011 & .012 & .009 \\ 
   \midrule
 \multicolumn{10}{l}{$\m{L}_\infty$-norm} \\
 \midrule
  0.1 & .045 & .046 & .047 & .047 & .037 & .036 & .050 & .041 & .042 & .038 & .060 \\ 
  0.05 & .044 & .043 & .047 & .045 & .045 & .039 & .051 & .042 & .043 & .041 & .042 \\ 
  0.01 & .008 & .011 & .012 & .015 & .009 & .010 & .011 & .011 & .008 & .009 & .019 \\ 
\bottomrule
\end{tabular*}
\begin{tablenotes}
\footnotesize
\item  \underline{Note}: See notes from Table \ref{tab:MC_estimation}.  Critical values based on \eqref{eq:g_p} for $p\in\{1,2,\infty\}$, $T_1=3$  and 10,000 replications. $\m{T}$ is  equally-spaced grid over $(0,1)$ containing 100 points.
\end{tablenotes}
\end{threeparttable}
\end{table}

 Overall, the test seems to be correctly sized in all scenarios considered.  Interestingly,  the size distortions for the Oracle are systematically larger than the ones observed in the Baseline, which might be partially explained by the increase in the estimation error shown in Panel A of Table \ref{thm:main}.  Also, the test based on the $\m{L}_\infty$ appears to be consistently slightly undersized, whereas the tests based on $\m{L}_1$ and $\m{L}_2$  are somewhat oversized in most cases.  Therefore,  all distributional tests seem satisfactory for practical purposes even in the Misspecified model considered (last column of Table \ref{T:test}).

\section{Empirical Illustration}\label{S:application}

To describe the methodology from the practitioner perspective, we revisit the empirical study in \citet*{acemoglu2016value},  henceforth AJKKM. Below, we discuss the basic setup, including a brief discussion of the authors' motivation and essential details about the dataset. For a comprehensive account, refer to AJKKM.

\subsection*{Background} On November 24, 2008, Tim Geithner was nominated by President-elect Obama as Treasury Secretary. The task is to understand the effects of the appointment announcement on the stock returns of financial firms supposedly connected to him. In the authors' words, \emph{In a time of crisis the position of Treasury Secretary has unusual powers, including over the financial system,  a point that was confirmed by the emergency legislation passed in October 2008. When immediate action is necessary, social connections are likely to become more important as sources of both ideas and human capital for the U.S. executive branch. In complex, stressful situations, it is natural to tap private sector friends, associates, and acquaintances with relevant expertise}. The authors find compelling evidence that firms linked to Geithner experienced higher cumulative abnormal returns (on average) than non-connected firms.

\subsection*{Dataset } The sample consists of 583 firms traded on NYSE or Nasdaq categorized as banks or financial services firms. The data consists of daily stock returns based on closing prices from 2007 and 2008. Out of those firms,   22 were considered treated, i.e., connected to Geithner. The firm's connection to Geithner was labeled based on three criteria: (A) Personal connection,  the number of shared board memberships between the firm's executives and Geithner ;  (B) Schedule Connections,  the number of times that Geithner interacted with executives from each firm while he was president of the New York Fed based on his schedule for each day from January 2007 through January 2009;  (C) N.Y. Headquartered,  a dummy variable equal to one if a firm's headquarters is identified as New York City since Geithner was president of the Federal Reserve Bank of New York. Those labels are not mutually exclusive.

Even though Geithner's nomination was officially announced on Monday, November 24,  news of his impending nomination was leaked to the press late in the trading day on Friday,  November 21, 2008, at approximately 3:00 p.m., which coincides with the beginning of a stock market rally. Therefore, we treat the first-day post-treatment ($T_0+1$ is our notation) as Friday,  November 21,  and refer to it as day $1$. Day  $2$ refers to the following Monday,  November 24, and so on throughout the subsequent trading days.

\subsection*{Results} We apply the methodology described in Section \ref{S:estimator} independently on each of the 22 connected firms using all the remaining $n=583-22$ firms as potential controls. The pre-intervention data are stock returns for 250 days ending 30 days before the announcement,  $T_0=250$. We consider the window of up to 10 days after the announcement, $T_1$. We take $\m{T}$ as an equally-spaced grid on $(0,1)$ with $1000$ points.  We present the results for one firm on each connection category with a single label.

The distribution counterfactual for selected firms in each connection type is shown in Figure \ref{fig:empirical_quant} together with the actual return for ten days after the announcement, while Figure \ref{fig:empirical_ci} plots the median and the 95\% confidence interval for announcement effect on returns for each of the ten post-announcement periods for the same firms selected in Figure \ref{fig:empirical_quant}. 

\begin{figure}[h]
\captionsetup[subfigure]{justification=centering}
\begin{subfigure}{0.32\textwidth}
    \centering
    \includegraphics[width=\textwidth]{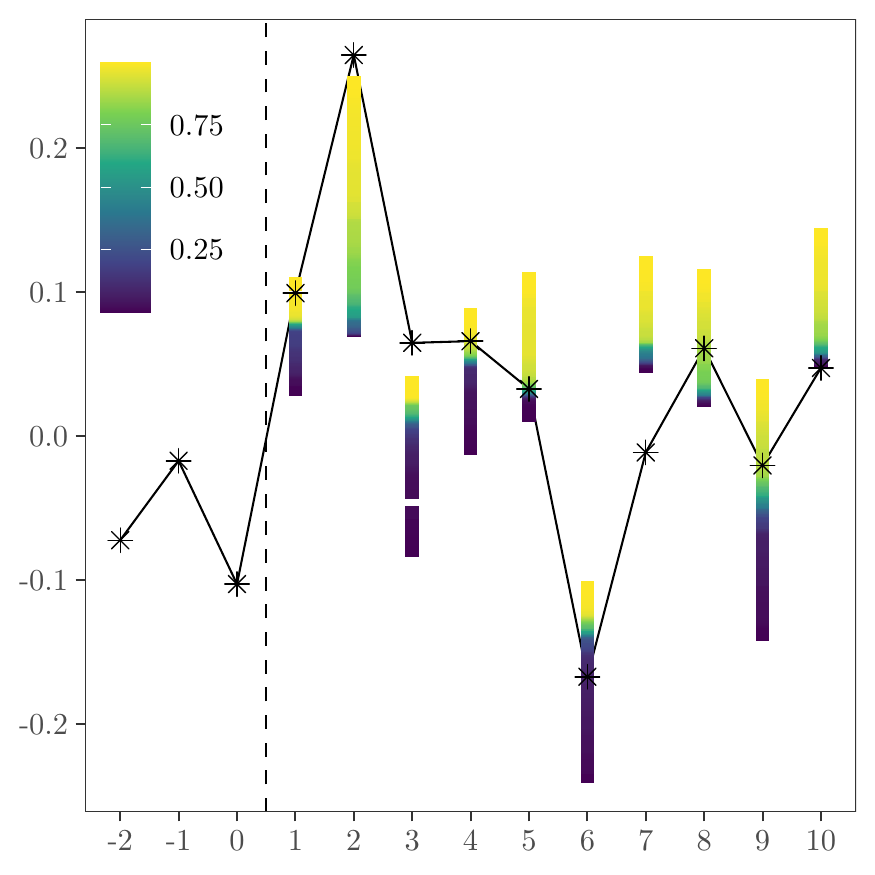}
    \caption{Personal Connection}
\end{subfigure}
\begin{subfigure}{0.32\textwidth}
    \centering
     \includegraphics[width=\textwidth]{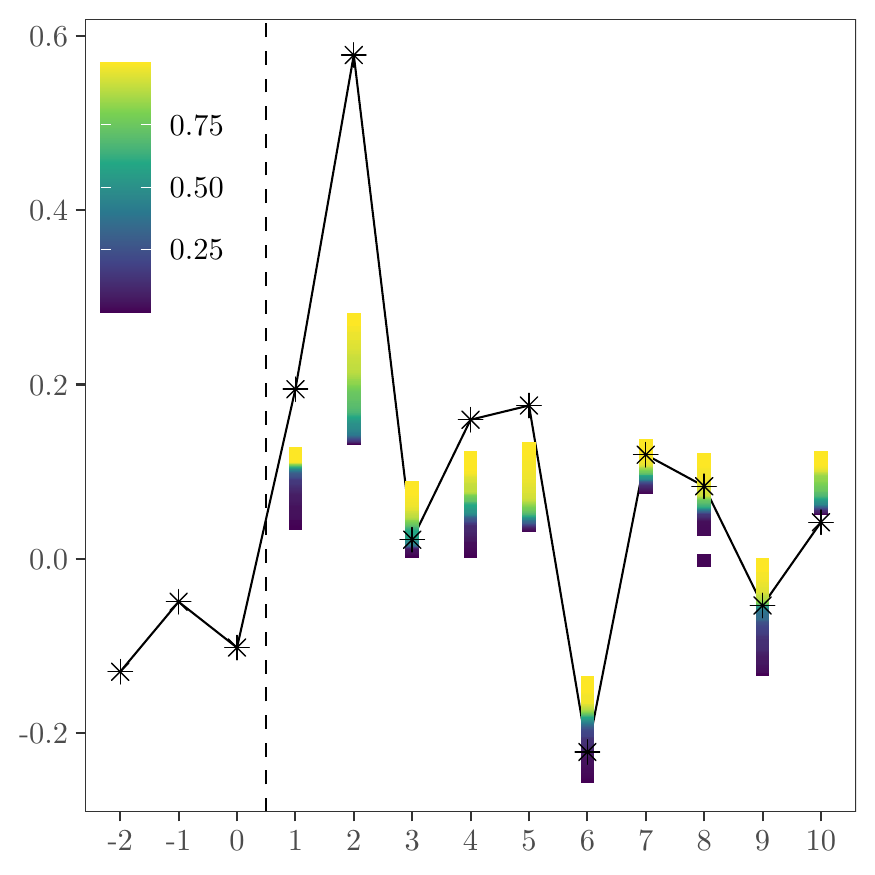}
	  \caption{Schedule Connection}
\end{subfigure}
\begin{subfigure}{0.32\textwidth}
    \centering
     \includegraphics[width=\textwidth]{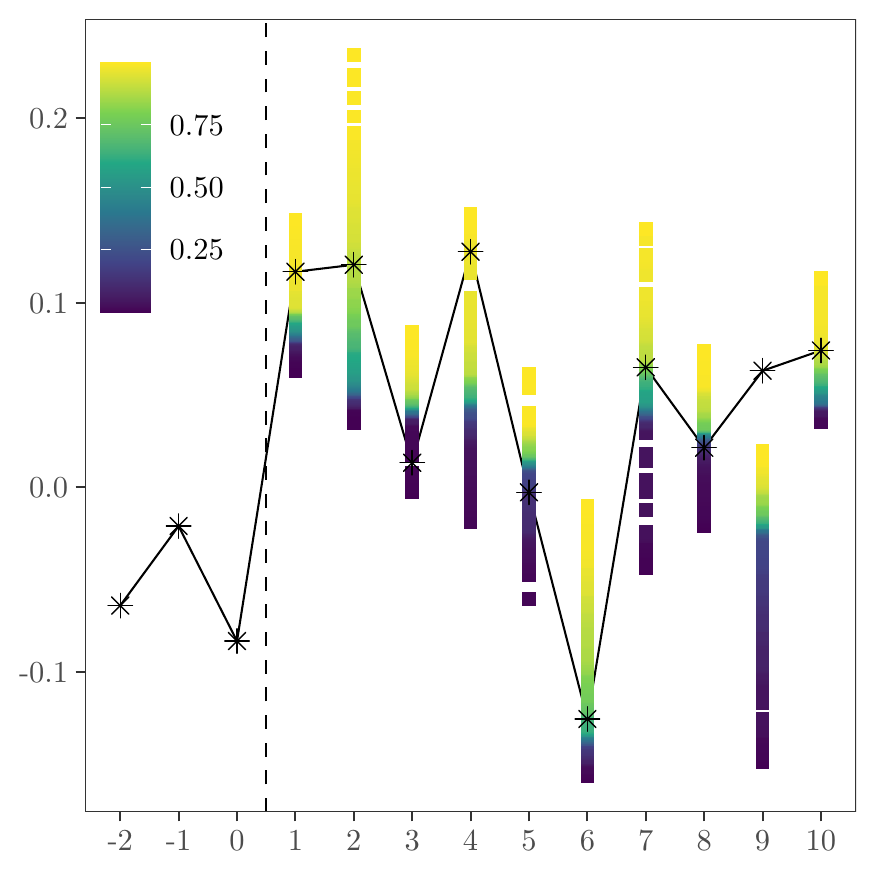}
	  \caption{NY Headquartered}
\end{subfigure}
\caption{Distributional Counterfactual Analysis for each type of connection. The solid black line is the actual returns. The color bar represents the conditional quantile given the peer's return. }\label{fig:empirical_quant}
\end{figure}

\begin{figure}[h]
\captionsetup[subfigure]{justification=centering}
\begin{subfigure}{0.32\textwidth}
    \centering
     \includegraphics[width=\textwidth]{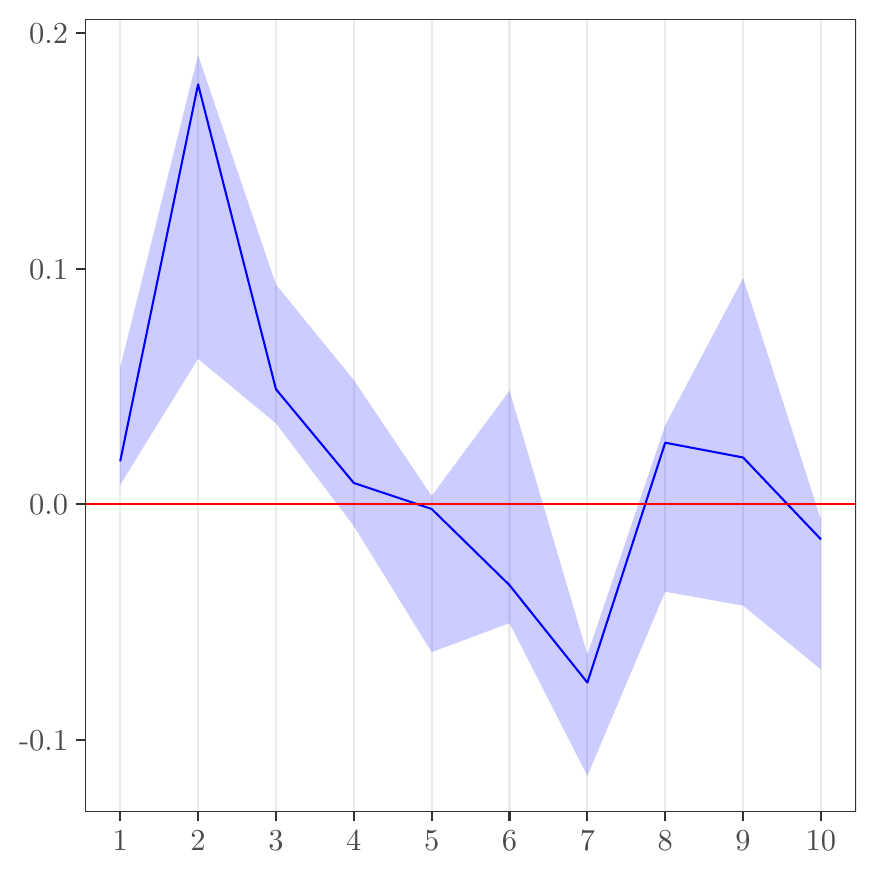}
    \caption{Personal Connection}
\end{subfigure}
\begin{subfigure}{0.32\textwidth}
    \centering
    \includegraphics[width=\textwidth]{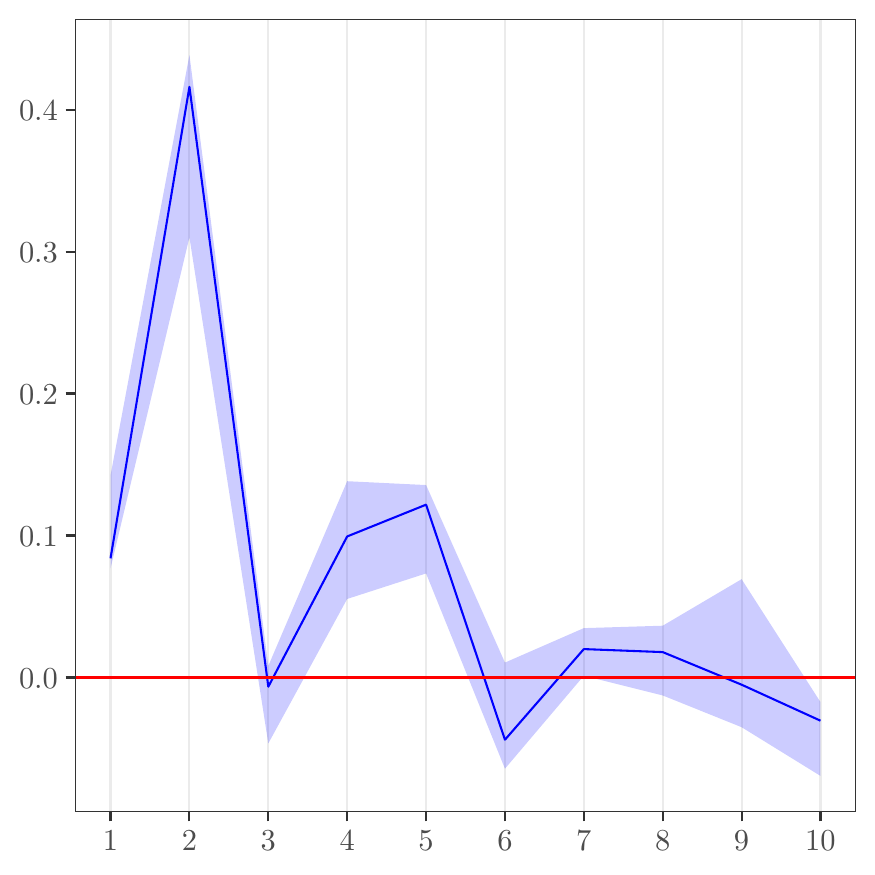}
	\caption{Schedule Connection}
\end{subfigure}
\begin{subfigure}{0.32\textwidth}
    \centering
	 \includegraphics[width=\textwidth]{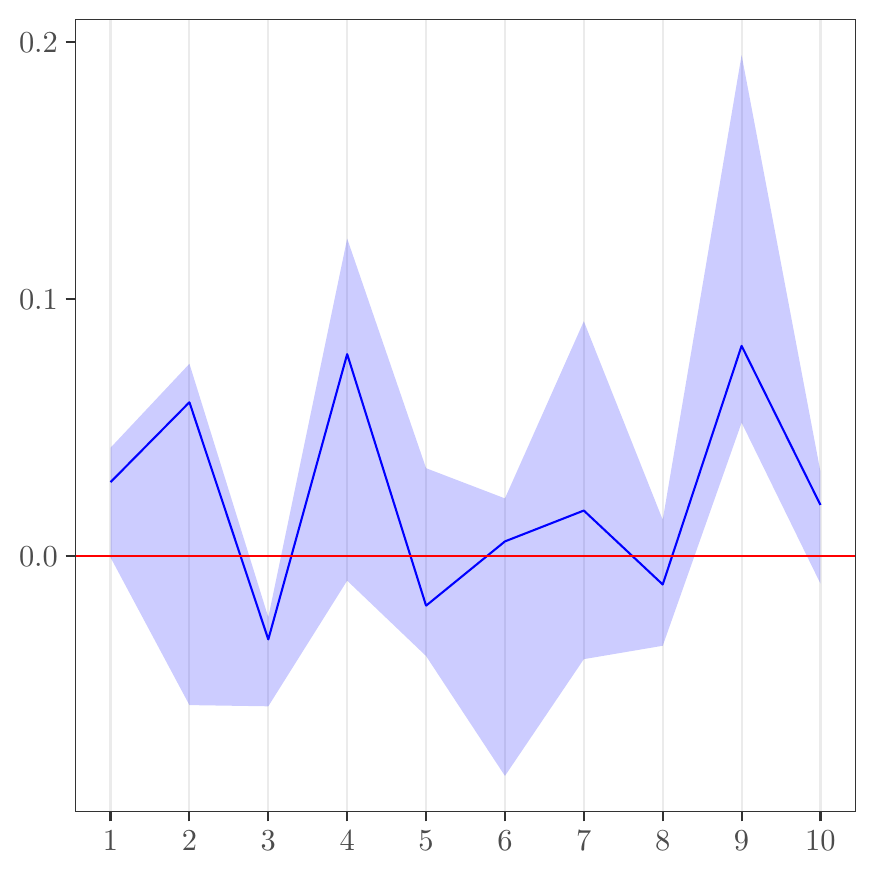}
	  \caption{NY Headquartered}
\end{subfigure}
\caption{Median and 95\% Confidence Interval for announcement effect on returns for each of the 10 post-announcement periods for the type of connection.}\label{fig:empirical_ci}
\end{figure}

Finally,  we test the null hypothesis of no effect over 2, 5, and 10 days after the announcement.  The null is tested individually for each of the 22 connected units using the test statistic \eqref{eq:test_statistic} for the norms $p\in\{1,2,\infty\}$. The p-value for each test is presented in Table \ref{tab:empirical}., and are computed via a Monte Carlo simulation using $10,000$ replications of  \eqref{eq:g_p} for each combination of $p$ and $T_1$.

In roughly half of the cases,  the null is rejected at $5\%$ across all the post-announcement windows and all $\m{L}_p$ norms considered. Overall, the rejection is more decisive for the shortest window $(T_1=2)$, which includes only the day of the leak and the official announcement. We have solid rejections for all connection types,  which indicates no systematic relation between the kind of connection and the effect on daily return.

Suppose we aggregate the results over all the connected units for each period. In that case, we conclude that the median daily return increase is about $3\%$ at day one,  around  $1\%$ at day two, and smaller than $0.3\%$ for the subsequent periods in a non-cumulative way. Interestingly, the distribution of the medians of the connected units at each given period is right-skewed, which might explain why AJKKM reports a \emph{mean} daily return increase of about $6\%$ after a full day of trading, which corresponds to day 1 in our analysis. We also have two other results using the same dataset as a comparison. The first is based on a penalized synthetic control method in \citet*{abadie2021}, and the second is obtained by applying a nonlinear factor model proposed in \citet*{feng2021causal}. The former reports a mean return increase of 6.1\% while the former reports an increase between $8$-$9\%$,  depending on the chosen parameter for the local PCA,  for the first full day of trading. 

In summary, based on these four independent analyses, there seems to be compelling evidence to support a positive effect on the daily returns for connected firms after Geithner was nominated Treasury Secretary. The degree of the impact depends on how we measure this effect (mean vs. median),  the aggregation level (all units vs. individual units), the type of connection, and the time span after the nomination (day one or cumulative over ten days).

\begin{table}[h]
\footnotesize
\centering
\begin{threeparttable}
\caption{p-values for the null hypothesis that the Geithner announcement did not affect the daily returns over $T_1$ periods after the announcement. \label{tab:empirical}}
\begin{tabular*}{\textwidth}{@{\extracolsep{\fill} } c c c c c c c c c c c c  }
\toprule
& & \multicolumn{3}{c}{$T_1 = 2$} & \multicolumn{3}{c}{$T_1 = 5$} & \multicolumn{3}{c}{$T_1 = 10$}\\
 \cmidrule{3-5} \cmidrule{6-8} \cmidrule{9-11}
  Unit ID &Connect & $p=1$ & $2$ & $\infty$ & $p=1$ & $2$ & $\infty$ & $p=1$ & $2$ & $\infty$\\
\midrule
23 & C & .347 & .369 & .491 & .216 & .247 & .344 & .043 & .046 & .077 \\ 
  40 & C & .768 & .769 & .738 & .824 & .812 & .742 & .086 & .096 & .182 \\ 
  51 & B & .039 & .051 & .065 & .598 & .675 & .856 & .141 & .100 & .097 \\ 
  58 & B,C & .000 & .000 & .000 & .054 & .050 & .030 & .046 & .042 & .031 \\ 
  68 & A & .000 & .000 & .000 & .226 & .230 & .315 & .017 & .018 & .012 \\ 
  73 & A,B & .504 & .506 & .505 & .404 & .353 & .253 & .296 & .299 & .249 \\ 
  74 & A,C & .002 & .002 & .003 & .152 & .186 & .274 & .053 & .065 & .103 \\ 
  122 & B & .000 & .000 & .000 & .002 & .002 & .001 & .035 & .016 & .005 \\ 
  131 & C & .379 & .396 & .505 & .434 & .487 & .545 & .216 & .281 & .520 \\ 
  192 & A & .001 & .001 & .002 & .007 & .003 & .001 & .006 & .005 & .015 \\ 
  197 & A & .390 & .404 & .505 & .254 & .282 & .315 & .141 & .160 & .216 \\ 
  261 & A & .000 & .000 & .000 & .008 & .006 & .005 & .047 & .046 & .028 \\ 
  309 & B,C & .665 & .666 & .588 & .544 & .503 & .378 & .373 & .396 & .364 \\ 
  323 & A & .096 & .103 & .114 & .025 & .036 & .073 & .384 & .316 & .207 \\ 
  356 & A,C & .000 & .000 & .000 & .093 & .072 & .030 & .016 & .016 & .009 \\ 
  363 & C & .023 & .027 & .032 & .162 & .158 & .143 & .061 & .058 & .086 \\ 
  379 & A & .215 & .253 & .342 & .074 & .083 & .110 & .002 & .002 & .009 \\ 
  385 & C & .578 & .579 & .532 & .084 & .068 & .031 & .175 & .150 & .067 \\ 
  386 & C & .067 & .079 & .095 & .545 & .480 & .334 & .458 & .461 & .324 \\ 
  428 & B & .175 & .214 & .288 & .467 & .452 & .353 & .074 & .050 & .014 \\ 
  441 & C & .338 & .363 & .505 & .091 & .101 & .115 & .022 & .019 & .047 \\ 
  491 & A & .002 & .002 & .003 & .018 & .013 & .012 & .029 & .028 & .021 \\ 
\bottomrule
\end{tabular*}
\begin{tablenotes}
\footnotesize
\item NB: \emph{Unit ID}: identifies the connected unit in the original dataset used by AJKKM. \emph{Connect} identifies the type of connection: (A) Personal Connection; (B) Schedule Connection; (C) if the unit's headquarters is in NYC. $T_1\in\{2,5,10\}$ is the number of post-announcement periods considered for the test. $p\in\{1,2,\infty\}$ is the $\m{L}_p$ considered to construct the test statistic \eqref{eq:test_statistic}. The reported p-values are simulated using 10,000 replications of \eqref{eq:g_p}.  For all cases, we consider $T_0=250$ pre-announcement periods and $n=561$ peers.
\end{tablenotes}
\end{threeparttable}
\end{table}

\section{Conclusion}\label{S:conclusion}

We propose a new methodology to carry out counterfactual analysis to evaluate the impact of interventions on the distribution of the variable of interest. Our approach is especially suited to situations when we observe a single (or just a few) treated unit and several potential controls. The setup is partially motivated by what is encountered in previous work. However, we depart from the standard approach of modeling only the conditional mean and model all the conditional quantiles. The methodology  also allows for the number of  units to be equal to or  much larger than the number of time periods, which is usually the case in event studies or when working with aggregated
data.

Therefore, we added to the counterfactual analysis toolkit. Moreover, the proposed methodology's estimation and test procedure can be quickly implemented using standard quantile estimation software and straightforward simulation\footnote{The Monte Carlo simulation results and the empirical application were coded in R based on the \emph{quantreg} package. The code is available from the author upon request.},  which is the reason we believe in its broad applicability to empirical work. As our empirical illustration demonstrates, virtually every empirical application that falls within the ``few treated units''  panel data setup can be re-considered under the proposed methodology to gain further insight via the treatment distributional effects. 

\newpage

\bibliographystyle{econ}

\bibliography{references}

\pagebreak

\begin{center}
\textbf{\large Supplemental Material: \textit{Distributional Counterfactual Analysis in High-Dimensional Setup}}
\end{center}
\vspace{1cm}

This supplemental material is organized into two appendices. Appendix \ref{App:proofs} contains the proofs of the main results, Theorem \ref{thm:main} and \ref{T:test} stated in Section \ref{S:Results} and its supporting lemmas. Appendix \ref{App:figures} collects additional figures from the empirical exercises described in Section \ref{S:application}.

\setcounter{equation}{0}
\setcounter{section}{0}
\setcounter{figure}{0}
\setcounter{table}{0}
\setcounter{page}{1}
\makeatletter
\renewcommand{\theequation}{S\arabic{equation}}
\renewcommand{\thefigure}{S\arabic{figure}}
\renewcommand{\bibnumfmt}[1]{[S#1]}
\renewcommand{\citenumfont}[1]{S#1}
\appendix

\section{Proofs}\label{App:proofs}

\subsection{Proof of Theorem \ref{thm:main}} 

The proof structure is as follows. The main argument relies on two key results. The first one is a determinist oracle-bound (Lemma \ref{lem:deterministic_part})  while the second control of the empirical process part (Lemma \ref{lem:control_empirical_process}). The latter, in turn, depends on the approximation of a beta-mixing sequence to an independent sequence (Lemma \ref{lem:mixing_to_independence}), which is based on Berbee's Coupling Lemma (refer to Proposition 2 in \citet*{DMR_1995}).

Before we start, let's prepare some additional notation. For $M,L,\lambda_0>0$ and $\m{T}$, define
\begin{align*}
    \m{E}&:=\m{E}_{M,\m{T}} := \sup|(\P_T - \P)(\rho_\tau(\theta) -\rho_\tau(\theta_0(\tau)))|\\
    \mathscr{A}&:=\mathscr{A}_{M,\m{T},L,\lambda_0}:=\{\m{E}_{M,\m{T}}\leq \lambda_0M, \max_{j,t}|X_{j,t}|\leq L\},
\end{align*}
where the supremum is taken over $\tau\in\m{T}$ and $\|\theta-\theta_0(\tau)\|_1\leq M$, and we use the empirical process standard notation $\P_T g = T_0^{-1}\sum_{t=1}^{T_0} g$, $\P =  T_0^{-1}\E\sum_{t=1}^{T_0} g $ for $\P$-measurable $g$ and $\rho_\tau(\theta) :=\rho_\tau(Y_t-X_t'\theta)$ for $1\leq t\leq T_0$ and $\theta\in\R^n$ for short.

For $\vartheta\geq 1$, set $L=\vartheta B \psi^{-1}( nT_0)$ and $\lambda_0=8\sqrt{2} \vartheta L\sqrt{a\log (4an)/T_0}$, then Lemma \ref{lem:control_empirical_process} yields
\[
\P\left(\m{E}> \lambda_0 M, \max_{j,t} |X_{t,j}|\leq L\right)\leq (4 an)^{1-\vartheta^2} + \frac{T_0}{a}\beta_a.
\]
The Markov inequality, a maximum inequality (Lemma 2.2.2 in \citet*{VW96}) and Assumption \ref{A:sampling_tail}(a) give us
\[
\P\left(\max_{t,j} |X_{t,j}|>L\right)\leq\frac{\E[\max_{j,t} |X_{t,j}|]}{L}\leq  \frac{C_\psi\max_{t,j}\|X_{t,j}\|_\psi}{\vartheta B}\leq \frac{C_\psi}{\vartheta},
\]
for some constant $C_\psi>0$ only depending on $\psi$.

Therefore, combining the last two displays and the union bound
\[
\P\left(\mathscr{A}^c\right)\leq(4 an)^{1-\vartheta^2} + \frac{T_0}{a}\beta_a + \frac{C_\psi}{\vartheta}.
\]
Set $M = M^*:=\frac{8\sqrt{2}\lambda^2s_0}{\lambda_0\varphi^2\underline{f}}$ and Lemma \ref{lem:deterministic_part} allow us to conclude that
\[
\P\left(\sup_{\tau\in\m{T}}\|\widehat{\theta}(\tau) - \theta_0(\tau)\|_1\gs \lambda s_0\right) \ls (an)^{1-\vartheta^2} + \frac{T_0}{a}\beta_a + \frac{1}{\vartheta}.
\]
provided that $\lambda\geq 8\lambda_0$ and $\frac{\lambda^2 s_0L}{\lambda_0}\leq\frac{3\underline{f}^2\varphi^2}{64\overline{f}}$.
Finally set $a=\lceil\beta^{-1}(1/(\vartheta T_0))\rceil$ to obtain \eqref{eq:main_thm_res_1} from the last expression since
\[
\P\left(\sup_{\tau\in\m{T}}\|\widehat{\theta}(\tau) - \theta_0(\tau)\|_1\gs\vartheta^2 B\psi^{-1}(nT_0)s_0\sqrt{\frac{(\beta^{-1}([\vartheta T_0)]^{-1})\lor 1)\log n}{T_0}}\right)  \ls n^{1-\vartheta^2} + \frac{1}{\vartheta}.
\]
 
For the second result, by the maximum inequality (Lemma 2.2.2 in \citet*{VW96}) and Assumption \ref{A:sampling_tail}(a),  $\|X_t\|_\infty\leq bB\psi^{-1}(n)$ with probability at least $1-1/b$ for any $b>0$. Then \eqref{eq:main_thm_res_2} follows because
\begin{align*}
\sup_{\tau\in\m{T}}|\h{Q}(\tau|X_t) - Q(\tau|X_t)|&\leq \sup_{\tau\in\m{T}}|\t{Q}(\tau|X_t) - Q(\tau|X_t)|\\
&=\sup_{\tau\in\m{T}}|X_t'\big[\widehat{\theta}(\tau)-\theta_0(\tau)\big]|\\
&\leq \|X_t\|_\infty \sup_{\tau\in\m{T}}\|\widehat{\theta}(\tau)-\theta_0(\tau)\|_1 \\
&\leq bB\psi^{-1}(n) \lambda s_0,
\end{align*}
with probability at least $ 1 - ( n^{1-\vartheta^2} + \frac{1}{\vartheta} + b^{-1})$. The first inequality follows from Proposition 4 in \citet*{chernozhukov2010quantile} with $p=\infty$, the second one from H{\"o}lder's inequality, and the last one by the first result and $\lambda$ set as stated in the Theorem.

For the last result, since $\tau =\P[Y_t(0)\leq Q(\tau|X_t)]$, we use Lemma \ref{L:useful} with $X$ equals to $U_t(\tau):=Y_t - \Q(\tau|X_t)$, $Y$ equals to $ \Q(\tau|X_t) -  \widehat{\Q}(\tau|X_t)$, $\eta=0$, and write for every $\tau\in\m{T}$ and $\delta>0$,
\[
|\P[Y_t\leq \widehat{Q}(\tau|X_t)] -\tau |\leq \P(|U_t(\tau)|\leq\delta) + \P(| Q(\tau|X_t) -  \h{Q}(\tau|X_t)|\geq \delta).
\]
Since $f(y|x)$ is bounded by Assumption \ref{A:conditional_density}(a), the first term is upper bounded by $2\delta\bar{f}$ uniformly in $\tau\in\m{T}$.  Set $\delta = \nu b\psi^{-1}(n)r$ and use the previous result to conclude the for all $\vartheta\geq 1$ and $b>0$ and some constant $C>0$
\[
|\P[Y_t\leq \widehat{Q}(\tau|X_t)] -\tau |\leq C\big[  n^{1-\vartheta^2} + \frac{1}{\vartheta}+ \frac{1}{b} + b B \lambda s_0 \psi^{-1}(n)\big].
\]
Taking the infimum of the right-hand size over $\vartheta\geq 1$ and $b>0$ concludes the proof of \eqref{eq:main_thm_res_3}.

 \begin{lemma}[\textbf{Deterministic Oracle Bound}]\label{lem:deterministic_part} On $\mathscr{A}^*:=\{\m{E}_{M^*,\m{T}}\leq \lambda_0M^*, \max_{j,t}|X_{j,t}|\leq L\}$ for $\lambda\geq 8\lambda_0$ and $\frac{\lambda^2 s_0L}{\lambda_0}\leq\frac{3\underline{f}^2\varphi^2}{64\overline{f}}$
\begin{equation}
\sup_{\tau\in\m{T}}\m{R}_\tau(\widehat{\theta})+\lambda\sup_{\tau\in\m{T}}\|\widehat{\theta}(\tau) - \theta_0(\tau)\|_1\leq \frac{64\lambda^2s_0}{\underline{f}\varphi^2},
\end{equation}
where $M^*:=\frac{8\sqrt{2}\lambda^2s_0}{\lambda_0\varphi^2\underline{f}}$ and $\m{R}_\tau(\theta):=\P (\rho_\tau(\theta) - \rho_\tau(\theta_0(\tau)))$ for $\theta\in\R^n$  and $\tau\in\m{T}$.
\end{lemma}
\begin{proof}[Proof of Lemma \ref{lem:deterministic_part}] The proof parallels the oracle inequalitiy proofs for general convex loss functions (refer to Ch. 6 of \citet*{buhlmann2011statistics}). The difference lies in the control of the bounds uniformly in $\tau\in \m{T}$. Define the convex combination $\widetilde{\theta}(\tau):= \alpha \widehat{\theta}(\tau) + (1-\alpha)\theta_0(\tau)$ with 
\[
\alpha := \frac{M^*}{M^*+\sup_{\tau\in\m{T}}\|\widehat{\theta}(\tau)-\theta_0(\tau)\|_1}\in[0,1].
\]
Then
\[
\sup_{\tau\in\m{T}}\|\widetilde{\theta}(\tau)-\theta_0(\tau)\|_1= \alpha\sup_{\tau\in\m{T}}\|\widehat{\theta}(\tau)-\theta_0(\tau)\|_1\leq M^*.
\]

Set $H_\tau(\theta):=\P_{T_0} \rho_\tau(\theta) + \lambda\|\theta\|_1$ for $\theta\in\R^n$ and $\tau\in\m{T}$. By the optimality of $\widehat{\theta}(\tau)$ we have $H_\tau(\widehat{\theta}(\tau))\leq H_\tau(\theta_0(\tau))$ for all $\tau\in\m{T}$. Moreover, since $H_\tau$ is convex
\[
H_\tau(\widetilde{\theta})\leq \alpha H_\tau(\widehat{\theta}(\tau)) + (1-\alpha)H_\tau(\theta_0(\tau))\leq H_\tau(\theta_0(\tau)).
\]
Which can be rewritten as
\[
\P (\rho_\tau(\widetilde{\theta}(\tau) - \rho_\tau(\theta_0(\tau))) + \lambda\|\widetilde{\theta}(\tau)\|_1\leq (\P_{T_0} -\P)[\rho_\tau(\theta_0(\tau)) - \rho_\tau(\widetilde{\theta}(\tau))] + \lambda\|\theta_0(\tau)\|_1.
\]
Then on the event $\mathscr{A}^*$ we have
\[
    \m{R}_\tau(\widetilde{\theta}(\tau)) + \lambda\|\widetilde{\theta}(\tau)\|_1\leq \lambda_0 M^* + \lambda\|\theta_0(\tau)\|_1.
\]
Decomposing $\theta = \theta_S -\theta_{S^c}$ for $S\subset\{1,\dots,n\}$ where $\theta_S$ is the vector $\theta$ with the entries not in $S$ set to zero.
\begin{equation}\label{eq:basic_oracle_inequality}
\m{R}_\tau(\widetilde{\theta}(\tau)) + \lambda\|\widetilde{\theta}_{S^c}(\tau)\|_1\leq \lambda_0 M^* + \lambda\|\widetilde{\theta}_{S}(\tau) - \theta_0(\tau)\|_1.
\end{equation}

We separate the analysis of the last expression into two cases: (i) if $\lambda_0 M^*\leq 2\lambda\|\widetilde{\theta}_{S}(\tau) - \theta_0(\tau)\|_1$ then
\[
\m{R}_\tau(\widetilde{\theta}(\tau)) + \lambda\|\widetilde{\theta}_{S^c}(\tau)\|_1\leq 3\lambda\|\widetilde{\theta}_{S}(\tau) - \theta_0(\tau)\|_1,
\]
hence $\widetilde{\theta}(\tau) - \theta_0(\tau)\in\m{C}_\tau$ for $\tau\in\m{T}$ and the compatibility condition (Assumption \ref{A:comp_cond}) give us $\|\widetilde{\theta}_{S}(\tau) - \theta_0(\tau)\|_1\leq \sqrt{s_0}\| \Sigma^{1/2}(\widetilde{\theta}(\tau) - \theta_0(\tau))\|_2/\varphi$. Also, a second-order Taylor expansion 
\begin{align*}
    \m{R}_\tau(\widetilde{\theta}(\tau)) &= \frac{1}{T_0}\sum_{t=1}^{T_0}\E\left\{ \frac{f(Q(\tau|X_t)|X_t)}{2}(X_t'(\widetilde{\theta}(\tau) - \theta_0(\tau))^2 \right.\\
    &\qquad+\left. \frac{\dot{f}(\widetilde{Y}_t|X_t)}{6}(X_t'(\widetilde{\theta}(\tau) - \theta_0(\tau))^3 ]\right\}\\
    &\geq \left[\frac{\underline{f}}{2} - \frac{\bar{f}}{6}\max_{t\leq T_0}\|X_t'(\widetilde{\theta}(\tau)-\theta_0(\tau))\|_\infty \right]\| \Sigma^{1/2}(\widetilde{\theta}(\tau) - \theta_0(\tau))\|_2^2,
\end{align*}
where we use Assumption \ref{A:conditional_density} to obtain the last inequality.

Then on $\sup_{\tau\in\m{T}}\max_{t\leq T_0}\|X'_t(\widetilde{\theta}(\tau)-\theta_0(\tau))\|_\infty \leq \eta$ we have $\inf_{\tau\in\m{T}}\m{R}_\tau(\theta)\geq \big[\frac{\underline{f}}{2} - \frac{\bar{f}\eta}{6}\big]\| \Sigma^{1/2}(\widetilde{\theta}(\tau) - \theta_0(\tau))\|_2^2 $. Also $\sup_{\tau\in\m{T}}\max_{t\leq T_0}\|X'_t(\widetilde{\theta}(\tau)-\theta_0(\tau))\|_\infty\leq \sup_{\tau\in\m{T}}\|\widehat{\theta}(\tau)-\theta_0(\tau)\|_1\max_{t,j}|X_{t,j}\|\leq M^*L$. So if set $\eta = \frac{3\underline{f}}{2\overline{f}}$ we have $M^*L\leq \eta$ because by assumption $\frac{\lambda^2 s_0L}{\lambda_0}\leq\frac{3\underline{f}^2\varphi^2}{64\overline{f}}$. Therefore $\inf_{\tau\in\m{T}}\m{R}_\tau(\theta)\geq \underline{f}/4 \| \Sigma^{1/2}(\widetilde{\theta}(\tau) - \theta_0(\tau))\|_2^2$ for $\theta$ such that $\sup_{\tau\in\m{T}}\|\theta - \theta_0(\tau)\|_1\leq M^*$. Using the fect that $H(u):=u^2/(4c)$ is the convex conjugate of $G(u) := c u^2$ for $c>0$ we have
\[
\frac{2\lambda\sqrt{s_0}\| \Sigma^{1/2}(\widetilde{\theta}(\tau) - \theta_0(\tau))\|_2}{\varphi}\leq \frac{4\lambda^2 s_0}{\kappa\varphi} + \frac{\m{R}_\tau(\widetilde{\theta}(\tau))}{2}.
\]
Then \eqref{eq:basic_oracle_inequality} yields
\[
\frac{\m{R}_\tau(\widetilde{\theta})}{2} + \lambda\|\widetilde{\theta}(\tau) - \theta_0(\tau)\|_1\leq \lambda_0 M^* + \frac{2\lambda\sqrt{s_0}\| \Sigma^{1/2}(\widetilde{\theta}(\tau) - \theta_0(\tau))\|_2}{\varphi}\leq  2\lambda_0 M^*.
\]

(ii) If $\lambda_0 M^*> 2\lambda\|\widetilde{\theta}_{S}(\tau) - \theta_0(\tau)\|_1$ then \eqref{eq:basic_oracle_inequality} implies that $\m{R}_\tau(\widetilde{\theta}(\tau)) + \lambda\|\widetilde{\theta}_{S^c}(\tau)\|_1\leq 3\lambda_0M^*$ and hence 
\[
\m{R}_\tau(\widetilde{\theta}(\tau)) + \lambda\|\widetilde{\theta}(\tau) - \theta_0(\tau)\|_1\leq 4\lambda_0M^*.
\]

Notice that in both cases the last two displays imply that $\|\widetilde{\theta}(\tau) - \theta_0(\tau)\|_1\leq 4\frac{\lambda_0}{\lambda} M^*\leq M^*/2$ since $\lambda\geq 8\lambda_0$ by assumption, which in turn implies that $\|\widetilde{\theta}(\tau) - \theta_0(\tau)\|_1\leq  M^*$. Therefore, we may repeat the arguments above with $\widetilde{\theta}(\tau)$ replaced by $\widehat{\theta}(\tau)$, which gives us the final result.
\end{proof}

\begin{lemma}[\textbf{Empirical Process Uniform Probability Bound}]\label{lem:control_empirical_process} Suppose that $\max_{1\leq t\leq T_0}\|X_t\|_\infty\leq L$ for some $L>0$, then for $x>0$ and $M>0$
\begin{equation}
\P(\m{E}>x)\leq \inf\left\{4 an\exp\left[-\left(\frac{x}{8\sqrt{2}ML}\sqrt{\frac{T_0}{a}}\right)^2\right] + \frac{T_0}{a}\beta_a\right\},
\end{equation}
where the infimum is taken over $a\in \{1,2,\dots, \lfloor n/2 \rfloor \}$.
\end{lemma}
\begin{proof}[Proof of Lemma \ref{lem:control_empirical_process}]
Recall that $\rho_\tau(\theta) :=\rho_\tau(Y_t-X_t'\theta)$ for $1\leq t\leq T_0$ and $\theta\in\R^n$ and $\m{E} := \sup|(\P_{T_0} - \P)(\rho_\tau(\theta) -\rho_\tau(\theta_0(\tau)))|$,
where the supremum is taken over $\tau\in\m{T}$ and $\|\theta-\theta_0(\tau)\|_1\leq M$.
From Lemma \ref{lem:mixing_to_independence} with  $\m{F}=\{(y,x)\to\rho_\tau(y-x'\theta):\theta\in\R^n, \|\theta-\theta_0(\tau)\|_1\leq M,\tau\in\m{T}\}$ we have, for $1\leq a\leq T_0$,
\begin{equation}\label{eq:union_bound}
    \P(\m{E}>x)\leq \P\left(\m{E}_{\mathsf{odd}}> x/2\right) + \P\left(\m{E}_{
\mathsf{even}}> x/2\right) + \frac{T_0}{a}\beta_a.
\end{equation}
Since both $\m{E}_{\mathsf{odd}}$ and $\m{E}_{\mathsf{even}}$ are sums of independent variables, by symmetrization of convex functions (Lemma 2.3.1 in \citet*{VW96}), for every $\zeta>0$,
\[
\E\exp(\zeta \m{E}_{\mathsf{odd}})\leq \E\exp(2\zeta \m{E}_{\mathsf{odd}}^0),
\]
where, for Radamacher random variables $\varepsilon_1,\dots \varepsilon_k$,
\[
\m{E}_{\mathsf{odd}}^0:=\sup\left|\sum_{k=1}^{m_1}\varepsilon_k \sum_{t\in H_{2k-1}}(\rho_\tau^t(\theta_0(\tau)) -\rho_\tau^t(\theta ))\right|. 
\]
It is easy to verify that $|\rho_\tau(u)-\rho_\tau(v)|\leq (\tau+1)|u-v|\leq 2|u-v|$. Therefore for $1\leq k \leq m_1$
\[
\left|\sum_{t\in H_{2k-1}}\rho_\tau^t(\theta_0(\tau)) -\rho_\tau^t(\theta ) \right|\leq 2\sum_{t\in H_{2k-1}}|X_t'(\theta-\theta_0(\tau))|.
\]
Define $\omega_k :=\left[\sum_{t\in H_k}\rho_\tau^t(\theta_0(\tau)) -\rho_\tau^t(\theta )\right]/\left[2\sum_{t\in H_k}|X_t'(\theta-\theta_0(\tau))|\right] $ with $0/0$ taken as $0$, then
\begin{align*}
    \E_\varepsilon\exp(2\zeta \m{E}_{\mathsf{odd}}^0) &= \E_\varepsilon\left\{\exp\left(4\zeta \sup\left|\sum_{k=1}^{m_1}\varepsilon_k\omega_k\sum_{t\in H_k}|X_t'(\theta-\theta_0(\tau))|\right|\right)\right\}\\
    &\leq \E_\varepsilon\left\{\exp\left(4\zeta \sup\sup_{\alpha_k\in[-1,1]}\left|\sum_{k=1}^{m_1}\varepsilon_k\alpha_k\sum_{t\in H_k}|X_t'(\theta-\theta_0(\tau))|\right|\right)\right\}\\
    &= \E_\varepsilon\left\{\exp\left(4\zeta \sup\left|\sum_{k=1}^{m_1}\varepsilon_k\sum_{t\in H_k}|X_t'(\theta-\theta_0(\tau))|\right|\right)\right\},
\end{align*}
where the inequality follows because $|\omega_k|\leq 1$. 

Let $\delta(\tau):=\theta-\theta_0(\tau)$ and the denote by $X_{j}^{(k)}$ the $j$-th observation of the $k$-th block, for $1\leq j\leq a$ and $1\leq k\leq m$. Apply triangle inequality to the right-hand side of the last expression to obtain
\begin{align*}
    \E_\varepsilon\exp(2\zeta \m{E}_{\mathsf{odd}}^0) &\leq  \E_\varepsilon\left\{\exp\left(4\zeta\sum_{j=1}^a\sup_{\|\delta(\tau)\|_1\leq M} \left|\sum_{k=1}^{m_1}\varepsilon_k|{X_j^{(k)}}'\delta(\tau)|\right|\right)\right\}.
\end{align*}
Notice that $\varepsilon_k|X_t'\delta(\tau)|$ has the same distribution as $\varepsilon_k X_t'\delta(\tau)$. Also, by H\"older's inequality 
\[
\left|\sum_{k=1}^{m_1}\varepsilon_k {X_j^{(k)}}'\delta(\tau)\right|\leq\left\|\delta(\tau)\right\|_1\left\|\sum_{k=1}^{m_1}\varepsilon_k X_j^{(k)}\right\|_\infty.
\]
Then
\begin{align*}
    \E_\varepsilon\exp(2\zeta \m{E}_{\mathsf{odd}}^0) &\leq  \E_\varepsilon\left\{\exp\left(4\zeta M\sum_{j=1}^a \left\|\sum_{k=1}^{m_1}\varepsilon_k X_j^{(k)}\right\|_\infty \right)\right\}\\
    &\leq n\max_{1\leq \ell\leq n}  \E_\varepsilon\left\{\exp\left(4\zeta M\sum_{j=1}^a \left|\sum_{k=1}^{m_1}\varepsilon_k X_{j,\ell}^{(k)}\right|\right)\right\}\\
      &\leq 2n\max_{1\leq \ell\leq n}  \E_\varepsilon\left\{\exp\left(4 \zeta M \sum_{k=1}^{m_1}\varepsilon_k \sum_{t\in H_{2k-1}} X_{t,\ell}\right)\right\}.
\end{align*}
Using the fact that $\exp(x) + \exp(-x)\leq 2\exp (x^2/2)$, we can bound the the last expectation by
\[\exp\left(\frac{(4\zeta M)^2}{2}\sum_{k=1}^{m_1} Y_{k,\ell}^2\right)= \exp\left(\frac{(4\zeta M\|Y_\ell\|_2)^2}{2}\right),\]
where $Y_{k,\ell}:=\sum_{t\in H_{2k-1}} X_{t,\ell}$ for $1\leq k\leq m_1$. Therefore
\[
\E_\varepsilon\exp(2\zeta \m{E}_{\mathsf{odd}}^0)\leq  2n\max_{1\leq \ell\leq n} 
\exp\left(\frac{(4\zeta M\|Y_\ell\|_2)^2}{2}\right)\leq 2n 
\exp\left(\frac{m(4\zeta a M L)^2}{2}\right).\]

By Markov's inequality, for any $x>0$
\begin{align*}
    \P(\m{E}_{\mathsf{odd}}>x)\leq \frac{\E\exp(\zeta \m{E}_{\mathsf{odd}})}{\exp(\zeta x)}\leq \frac{\E\E_\varepsilon\exp(2\zeta \m{E}^0_{\mathsf{odd}})}{\exp(\zeta x)}\leq 2n
\exp\left(\frac{m(4\zeta a ML)^2}{2} - \zeta x\right).
\end{align*}
Set $\zeta =\frac{x}{16m(aML)^2} $ to conclude
\[\P(\m{E}_{\mathsf{odd}}>x)\leq 2an\exp\left[-\frac{1}{2m}\left(\frac{xT_0}{4aML}\right)^2\right] =2an\exp\left[-\frac{1}{2}\left(\frac{x\sqrt{T_0/a}}{4ML}\right)^2\right]. \]
The same upper bound can be derived for $\m{E}_\mathsf{even}$. Therefore for $a\in\{1,\dots, \lfloor n/2 \rfloor \}$, $M,x>0$ it follows from \eqref{eq:union_bound}
\[
\P(\m{E}>x)\leq 4 an\exp\left[-\frac{1}{2}\left(\frac{x\sqrt{T_0/a}}{8ML}\right)^2\right] + \frac{T_0}{a}\beta_a.
\]
Take the infimum to conclude.
\end{proof}

\begin{lemma}[\textbf{From $\beta$-mixing to independence}]\label{lem:mixing_to_independence} For $T\in\N$, let $\{X_t:1\leq t\leq T\}$ be a sequence of random variables taking value on a Polish space $\m{X}$ with $\beta$-mixing coefficient denoted by $\{\beta_n\}_{n\in\N}$. Also, let $\m{F}$ be a class of measurable real-valued functions on $\m{X}$. Then, for $x>0$,
\begin{align*}
     \P\left(\sup_{f\in\m{F}}|(\P_T -\P) f|\geq x\right)&\leq \inf_{1\leq a\leq T}\left\{\P\left(\sup_{f\in\m{F}}|\Q_{\textsf{odd}} f|> x/2\right)\right.\\
      &\left.\qquad+\P\left(\sup_{f\in\m{F}}|\Q_{\textsf{even}} f|> x/2\right) + \frac{T}{a}\beta_a\right\},
\end{align*}
with 
\begin{align*}
   \Q_{\textsf{odd}} f &:= \tfrac{1}{T}\sum_{k=1}^{m_1} \sum_{t\in H_{2k-1}} f(\widetilde{X}_t)- \E f(\widetilde{X}_t)\\
     \Q_{\textsf{even}} f &:= \tfrac{1}{T}\sum_{k=1}^{m_2} \sum_{t\in H_{2k}} f(\widetilde{X}_t)- \E f(\widetilde{X}_t),
\end{align*}
where $m := \lceil T/a\rceil$ and $m_1 + m_2 = m$, $H_k := \{t\in \N :(k-1)a <t \leq ka\land T\}$ for $1\leq k\leq m$ and the sequence $\{\widetilde{X}_1,\dots, \widetilde{X}_{T}\}$ is  such that the distribution of $\{\widetilde{X}_t:t\in H_k\}$ is equal to the distribution of $\{X_t:t\in H_k\}$ for every $1\leq k\leq m+1$, and $\{\widetilde{X}_t:t\in H_k\}$ is independent of $\{\widetilde{X}_t:t\in H_j\}$ for $k\neq j$. We define summation over an empty set as zero.
\end{lemma}

\begin{proof} For $f\in\m{F}$ and $1\leq k\leq m$, define the random variables
\begin{align*}
   Y_{f,k} := \sum_{t\in H_{2k-1}} f(X_t) -\E f(X_t);\qquad Z_{f,k} := \sum_{t\in H_{2k}} f(X_t) -\E f(X_t).
\end{align*}
Then, by construction
\begin{align*}
    (\P_T - \P) f :=\frac{1}{T}\sum_{t=1}^T f(X_t) -\E f(X_t) &= \frac{1}{T}\sum_{k=1}^{m_1} Y_{f,k} + \frac{1}{T}\sum_{k=1}^{m_2} Z_{f,k}.
\end{align*}
Therefore, for $x>0$, we have by the union bound
\begin{align*}
    \P\left(\sup_{f\in\m{F}}|(\P_T - \P) f|\geq x\right) &\leq \P\left(\sup_{f\in\m{F}}\left|\frac{1}{T}\sum_{k=1}^{m_1}Y_{f,k}\right|\geq x/2\right) + \P\left(\sup_{f\in\m{F}}\left|\frac{1}{T}\sum_{k=1}^{m_2}Z_{f,k}\right|\geq x/2\right).
\end{align*}

Now construct another sequence $\{\widetilde{X}_1,\dots, \widetilde{X}_{T}\}$ as described in the Lemma. Berbee's Coupling Lemma ensures the construction (refer to Proposition 2 in \citet*{DMR_1995}). Finally, define $\widetilde{Y}_{f,k} := \sum_{t\in H_{2k-1}} f(\widetilde{X}_t) - \E f(\widetilde{X}_t)$ and $\widetilde{Z}_{f,k} := \sum_{t\in H_{2k}} f(\widetilde{X}_t) - \E f(\widetilde{X}_t)$ for $f\in\m{F}$ and $1\leq k\leq m$. 

Note that by construction  $\{\widetilde{Y}_{f,k}:1\leq k\leq m_1\}$ is an independent sequence such that $\widetilde{Y}_{f,k}$ and $Y_{f,k}$ share the same distribution for all $1\leq k\leq m$ and $f\in\m{F}$. Also, since the elements of sequence $\{Y_{f,k}:1\leq k\leq m\}$ are separated by $a$ (with respect to the original sequence $\{X_t:1\leq t\leq T$), we have that the $\beta$-mixing coefficient of $\{Y_{f,k}:1\leq k\leq m\}$ is upper bound by $\beta_a$ so Berbee's Coupling Lemma also ensures that $\P(Y_{f,k}\neq \widetilde{Y}_{f,k}) \leq \beta_a$ for all $1\leq k\leq m_1$ and $f\in\m{F}$. 

We now bound the difference of functions on these two sequences. For a measurable $g:\R^{m+1}\to \R$ we have by substitution
\[
|g(Y_{f,1},\dots,Y_{f,m+1}) - g(\widetilde{Y}_{f,1},\dots,\widetilde{Y}_{f,m_1})|\leq \|g\|_\infty\sum_{k=1}^{m_1}\1\{Y_{f,k}\neq \widetilde{Y}_{f,k}\}.
\]
Use this inequality with $g$ being the indicator function of the event $\{\sup_{f\in\m{F}}\sum_{k=1}^m Y_{f,k}>x/2\}$ and the fact that
\[
|\E g(Y_{f,1},\dots,Y_{f,m+1}) - \E g(\widetilde{Y}_{f,1},\dots,\widetilde{Y}_{f,m_1})|\leq \E| g(Y_{f,1},\dots,Y_{f,m+1}) - g(\widetilde{Y}_{f,1},\dots,\widetilde{Y}_{f,m_1})|.
\]
to conclude
\begin{align*}
    \left|\P\left(\sup_{f\in\m{F}}\left|\frac{1}{T}\sum_{k=1}^{m_1}Y_{f,k}\right|\geq x/2\right) - \P\left(\sup_{f\in\m{F}}\left|\frac{1}{T}\sum_{k=1}^{m_1}\widetilde{Y}_{f,k}\right|\geq x/2\right)\right|&\leq m_1\beta_a.
\end{align*}
Following the same steps above  for the sequences $\{Z_{f,k}:1\leq k\leq {m+1}\}$ and $\{\widetilde{Z}_{f,k}:1\leq k\leq {m+1}\}$ we have
\begin{align*}
     \left|\P\left(\sup_{f\in\m{F}}\left|\frac{1}{T}\sum_{k=1}^{m_2}Z_{f,k}\right|\geq x/2\right) - \P\left(\sup_{f\in\m{F}}\left|\frac{1}{T}\sum_{k=1}^{m_2}\widetilde{Z}_{f,k}\right|\geq x/2\right)\right|&\leq m_2\beta_a.
\end{align*}
The result then follows because $m_1 + m_2 = m$.
\end{proof}

\subsection{Proof of Theorem \ref{T:test}}

Fix a compact $\m{T}\subset (0,1)$ and $p\in[0,1]$ and let $t\in\R$ be a continuity point of $x\mapsto G_p(x):= \P(\phi_{\m{T},p}^0\leq x)$. Applying Lemma \ref{L:useful},  we have for any $\delta>0$,
\begin{equation}\label{eq:split_theorem2}
|\P(\phi_{\m{T},p}\leq t) - \P(\phi_{\m{T},p}^0\leq t)| \leq |\P(|\phi_{\m{T},p}^0-t|\leq \delta)|  + \P(|\phi_{\m{T},p}-\phi_{\m{T},p}^0| >\delta).
\end{equation}
When $\m{T}$ is countable,  $G_p(x)$  is constant around some small enough neighborhood of $t$; otherwise, it has bounded density in its continuous portion. In both cases,  the first term in \eqref{eq:split_theorem2} is bounded by $C\delta$ for some constant $C$ only depending on $p$.  Thus, we are left to bound the second term.

Apply triangle inequality twice followed by Holders inequality to write
\begin{align*}
   |\phi_{\m{T},p}-\phi_{\m{T},p}^0|&\leq \frac{1}{T_1}\sum_{t>T_0}^T\left|\|\1\{Y_t\leq \widehat{Q}(\tau|X_t)\}\|_{p,\m{T}} - \|\1\{Y_t\leq Q(\tau|X_t)\}\|_{p,\m{T}}\right|\\
   &\leq \frac{1}{T_1}\sum_{t>T_0}^T\|\1\{Y_t\leq \widehat{Q}(\tau|X_t)\} - \1\{Y_t\leq Q(\tau|X_t)\}\|_{p,\m{T}}\\
   &\leq \frac{1}{T_1}\sum_{t>T_0}^T\|\1\{Y_t\leq \widehat{Q}(\tau|X_t)\} - \1\{Y_t\leq Q(\tau|X_t)\}\|_{\infty,\m{T}}.
\end{align*}

Recall that $\1\{Y_t\leq \h{Q}(\tau|X_t)\} = \1\{Y_t -\Delta_t(\tau)\leq Q(\tau|X_t)\} = \1\{F(Y_t -\Delta_t(\tau))|X_t)\leq\tau\}=  \1\{U_t -Z_t(\tau)\leq\tau\}$ where $\Delta_t(\tau):=\widehat{Q}(\tau|x)-Q(\tau|x)$ and  $Z_t(\tau):=-f(\t{Z}(\tau))\Delta_t(\tau)$ for some $\t{Z}_t(\tau)$ by the mean value theorem using Assumption \ref{A:conditional_density}(a).  Since $f$ is uniformly bounded by the \ref{A:conditional_density}(a) we have
\begin{align*}
    \|\1\{Y_t\leq \widehat{Q}(\tau|X_t)\} - \1\{Y_t\leq Q(\tau|X_t)\}\|_{\infty,\m{T}} &= \|\1\{U_t - Z_t(\tau)\leq \tau\} - \1\{U_t \leq \tau\}\|_{\infty,\m{T}}\\
    &\leq \overline{f}\sup_{\tau\in\m{T}}|\Delta_t(\tau)|.
\end{align*}
Set $\delta = bB \lambda s_0 \psi^{-1}(n)$ for $\vartheta\geq 1$ and $b>0$; and by the second result in Theorem \ref{thm:main}, we obtain
\begin{align*}
    \P( |\phi_{\m{T},p}-\phi_{\m{T},p}^0|>\delta)\leq \sum_{t>T_0}^T\P(\sup_{\tau\in\m{T}}|\Delta_t(\tau)|>\delta/\overline{f})\ls T_1\left( n^{1-\vartheta^2} + \frac{1}{\vartheta}+ \frac{1}{b}\right).
\end{align*}
The result follows from \eqref{eq:split_theorem2} and the above expression.

\begin{lemma}\label{L:useful} Let $X,Y$ be random variables then for $\eta\in \R$ and $\delta\geq 0$:
\[|\P(X+Y\leq \eta) -\P(X\leq\eta)|\leq \P(|X-\eta|\leq\delta) + \P(|Y|\geq \delta) \]
\end{lemma}
\begin{proof} For $\delta>0$ we have
\begin{align*}
\P(X+Y\leq \eta) &\leq \P(X+Y\leq\eta \cap |Y|\leq \delta) +  \P( |Y|> \delta)\\
&\leq \P(X\leq\eta  +\delta) +  \P( |Y|> \delta)\\
&= \P(X\leq\eta) + \P(\eta<X\leq\eta  +\delta)  +  \P( |Y|> \delta),
\end{align*}
and 
\begin{align*}
\P(X+Y\leq \eta) &\geq \P(X\leq\eta-\delta \cap |Y|\leq \delta)\\
&= \P(X\leq\eta  -\delta)  + \P(X\leq\eta-\delta \cap |Y|\leq \delta) -\P(X\leq\eta  -\delta)\big]\\
&= \P(X\leq\eta) -  \big[\P(\eta-\delta <X\leq\eta) + \P(X\leq\eta  -\delta)- \P(X\leq\eta-\delta \cap |Y|\leq \delta)\big]\\
&\geq \P(X\leq\eta) -  \big[\P(\eta-\delta <X\leq\eta) + \P(|Y|> \delta)\big].
\end{align*}
Hence
\begin{align*}
-  \big[\P(\eta-\delta <X\leq\eta) + \P(|Y|> \delta)\big]\leq      &\\
\P(X+Y\leq \eta) -\P(X\leq\eta) &\leq \\
\P(\eta<X\leq\eta  +\delta)&  +  \P( |Y|> \delta).
\end{align*}
Using that fact that $-a\leq x \leq b$ implies  $|x|\leq a\lor b$ for $a,b\geq 0$, we conclude
\begin{align*}
|\P(X+Y\leq \eta) -\P(X\leq\eta)|&\leq \big[\P(\eta-\delta <X\leq\eta)\lor \P(\eta<X\leq\eta  +\delta)\big]\\ &\qquad +  \P( |Y|> \delta)\\
&\leq \P(|X-\eta|  \leq \delta)+  \P( |Y|> \delta).
\end{align*}
\end{proof}

\newpage
\section{Additional Figures}\label{App:figures}

\begin{figure}[h!]
\captionsetup[subfigure]{justification=centering, labelformat=empty}
\begin{subfigure}{0.32\textwidth}
    \centering
    \includegraphics[width=\textwidth]{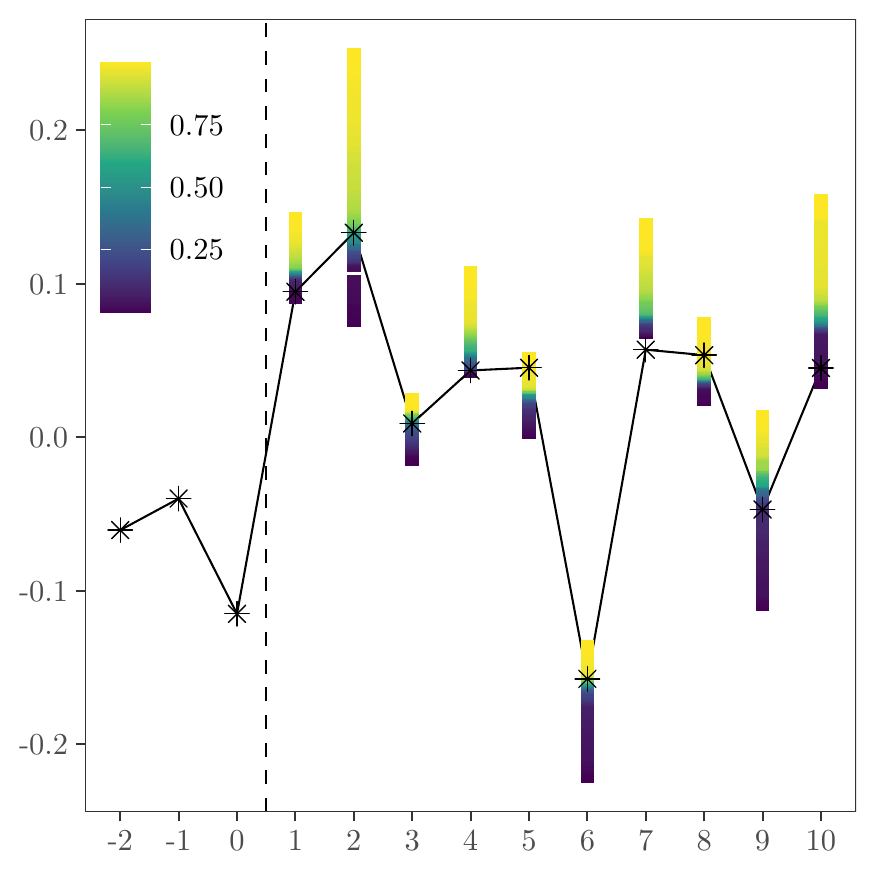}
    \caption{Personal Connection}
\end{subfigure}
\begin{subfigure}{0.32\textwidth}
    \centering
     \includegraphics[width=\textwidth]{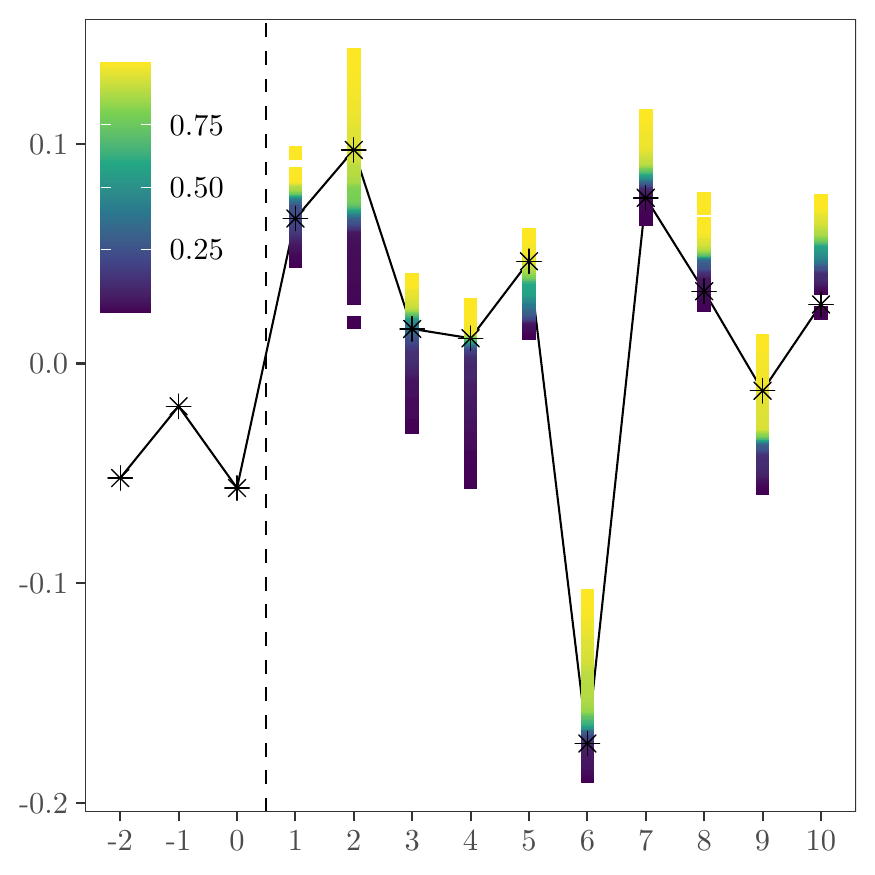}
	  \caption{Schedule Connection}
\end{subfigure}
\begin{subfigure}{0.32\textwidth}
    \centering
     \includegraphics[width=\textwidth]{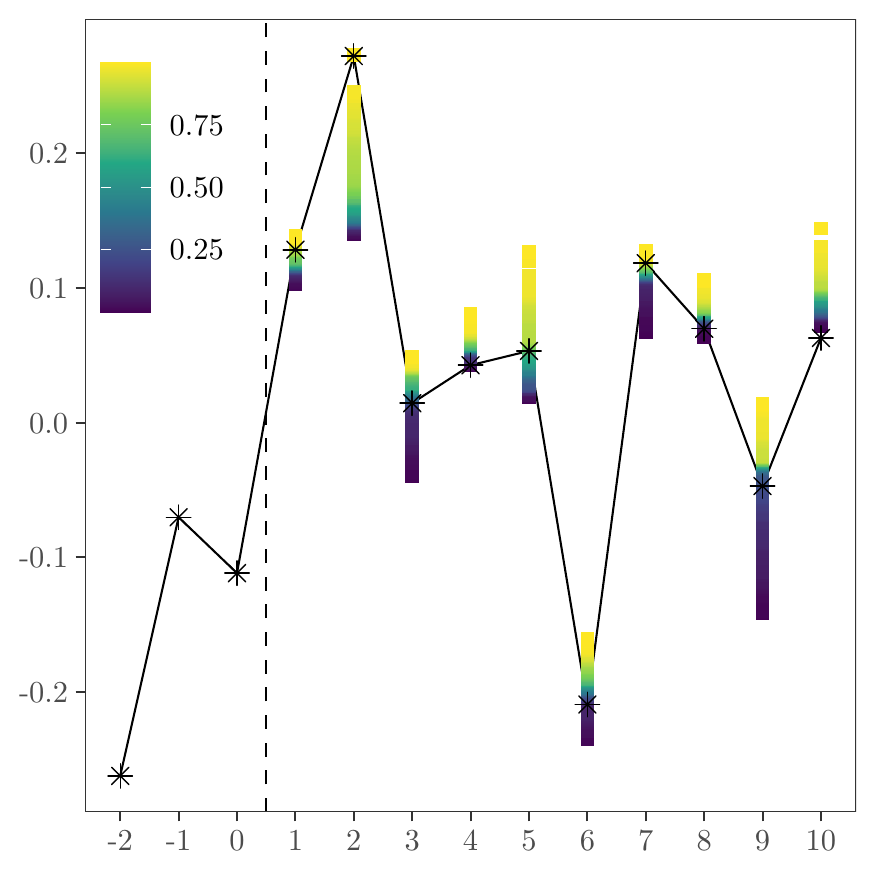}
	  \caption{NY Headquartered}
\end{subfigure}
\begin{subfigure}{0.32\textwidth}
    \centering
    \includegraphics[width=\textwidth]{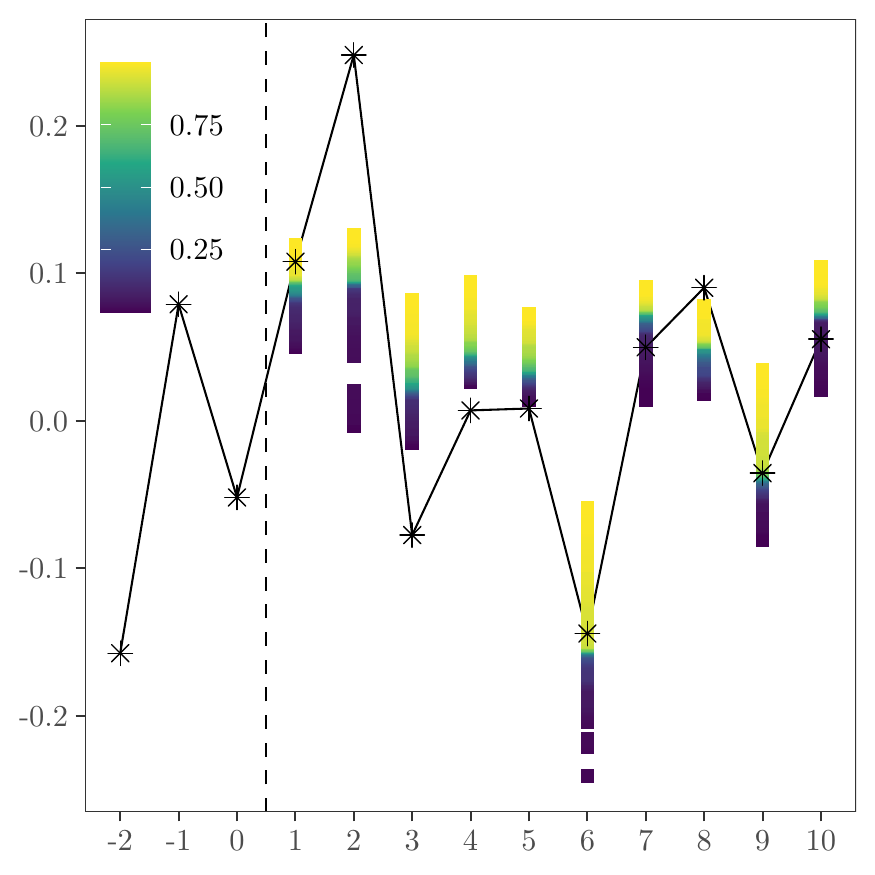}
    \caption{Personal Connection}
\end{subfigure}
\begin{subfigure}{0.32\textwidth}
    \centering
     \includegraphics[width=\textwidth]{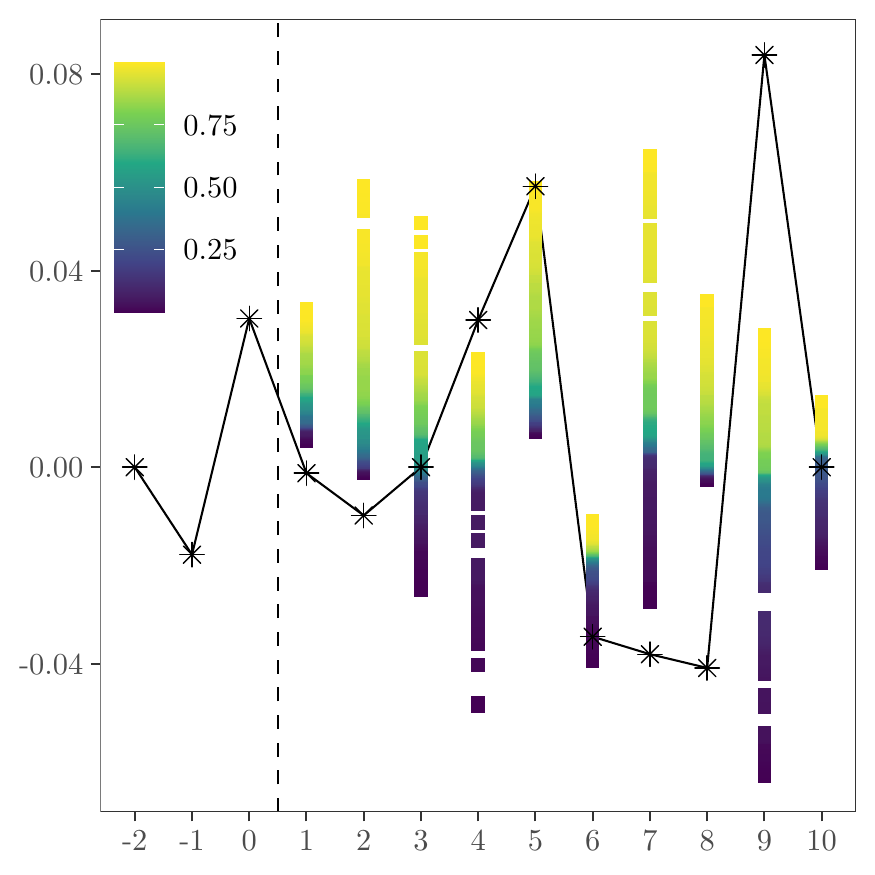}
	  \caption{Schedule Connection}
\end{subfigure}
\begin{subfigure}{0.32\textwidth}
    \centering
     \includegraphics[width=\textwidth]{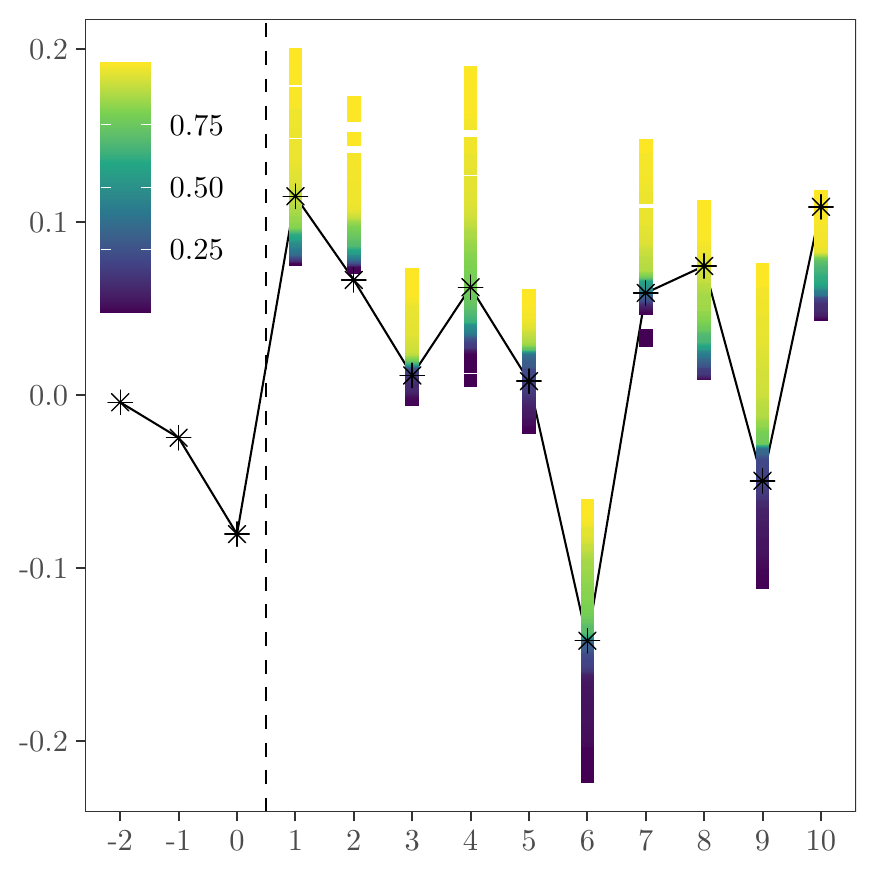}
	  \caption{NY Headquartered}
\end{subfigure}
\begin{subfigure}{0.32\textwidth}
    \centering
    \includegraphics[width=\textwidth]{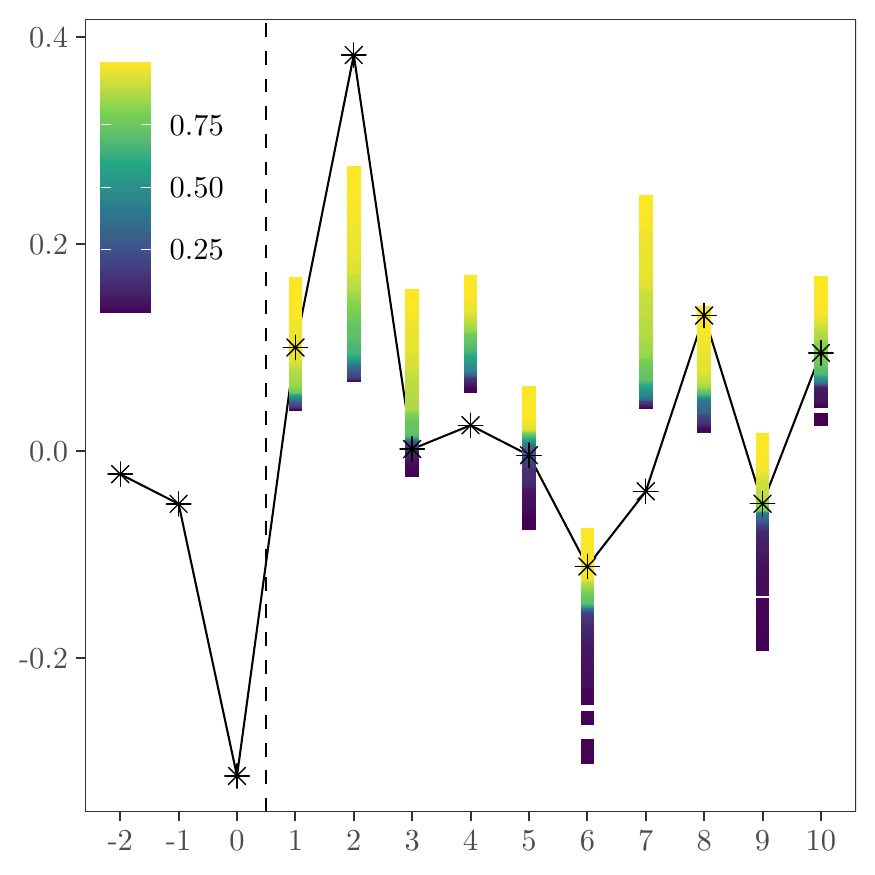}
    \caption{Personal Connection}
\end{subfigure}
\begin{subfigure}{0.32\textwidth}
    \centering
     \includegraphics[width=\textwidth]{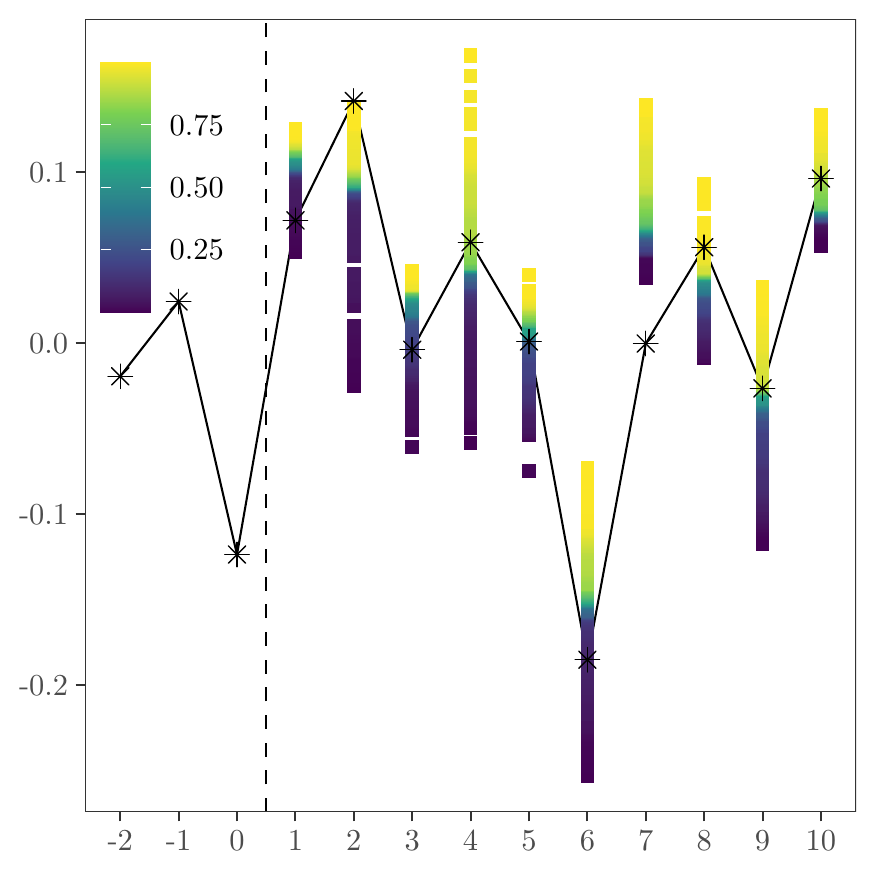}
	  \caption{Schedule Connection}
\end{subfigure}
\begin{subfigure}{0.32\textwidth}
    \centering
     \includegraphics[width=\textwidth]{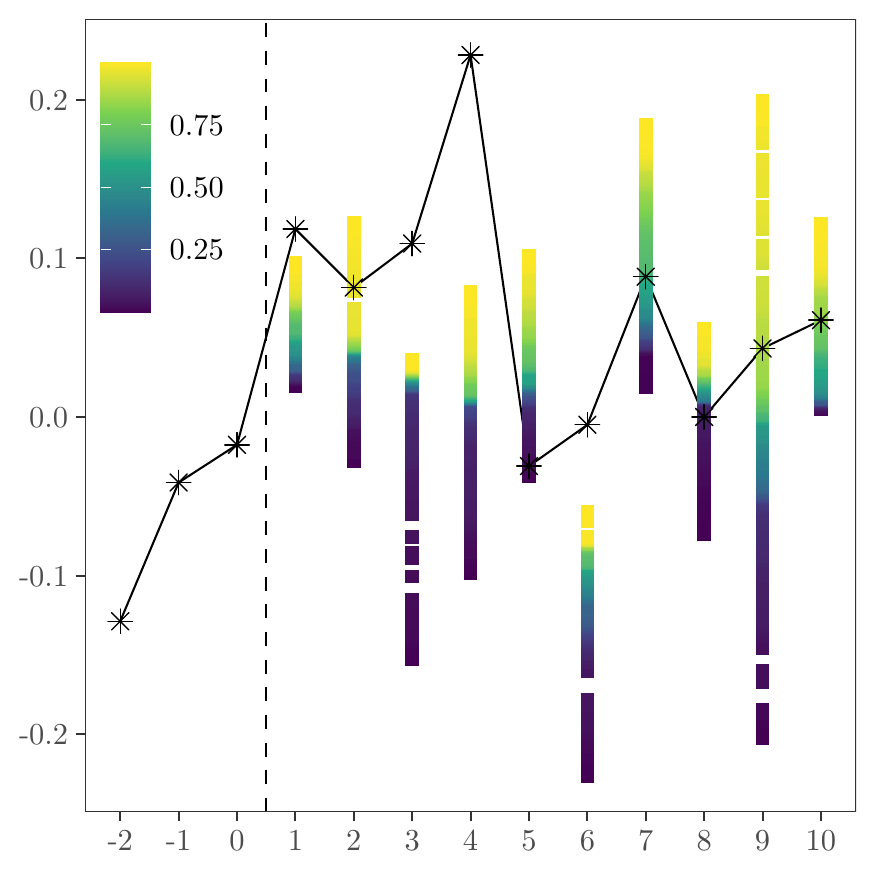}
	  \caption{NY Headquartered}
\end{subfigure}
\caption{Distributional Counterfactual Analysis for each type of connection.  The solid black line is the actual returns.  The color bar represents the conditional quantile given the peer's return. }\label{fig:empirical_quant_appendix}
\end{figure}

\begin{figure}[h!]
\captionsetup[subfigure]{justification=centering, labelformat=empty}
\begin{subfigure}{0.32\textwidth}
    \centering
     \includegraphics[width=\textwidth]{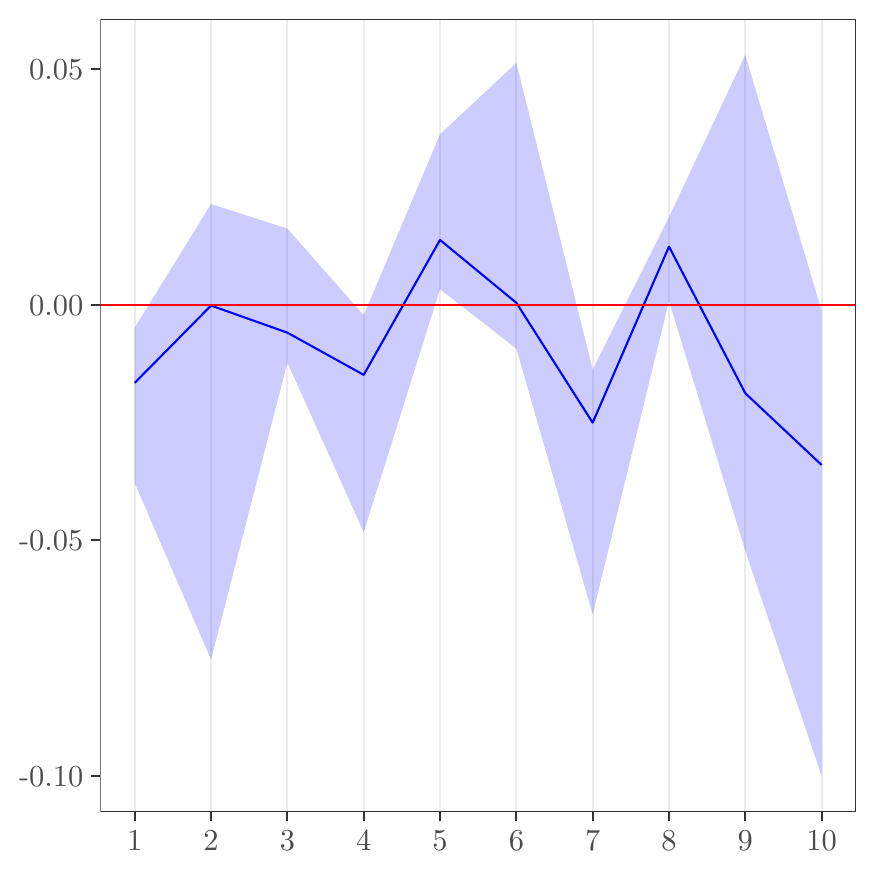}
    \caption{Personal Connection}
\end{subfigure}
\begin{subfigure}{0.32\textwidth}
    \centering
    \includegraphics[width=\textwidth]{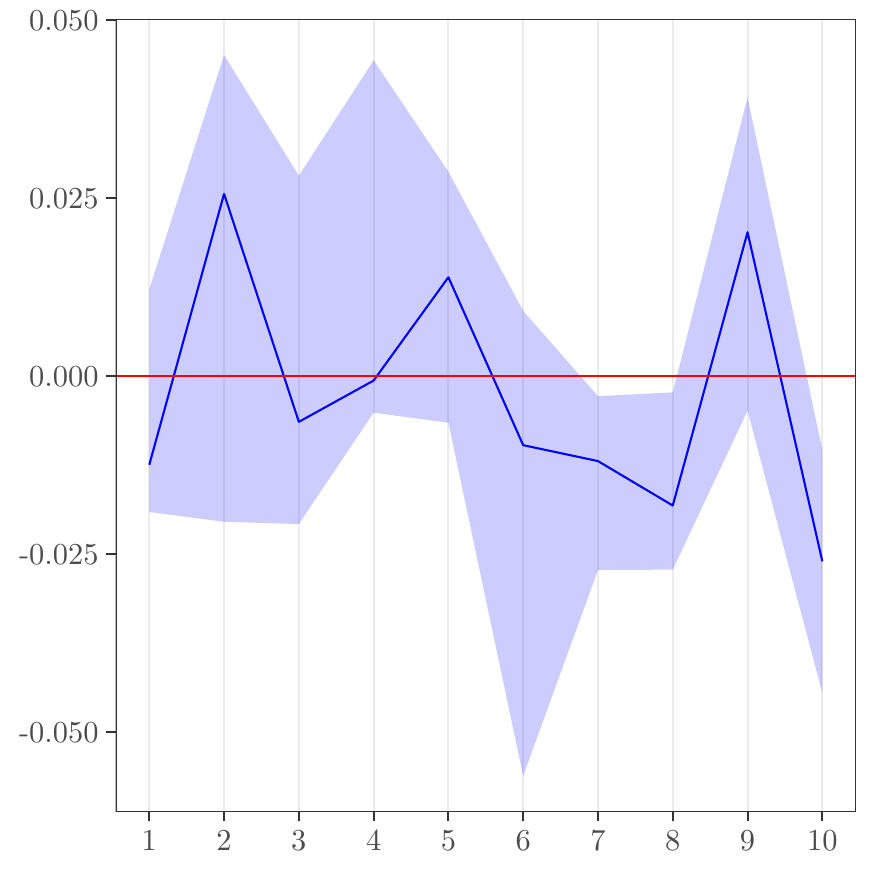}
	\caption{Schedule Connection}
\end{subfigure}
\begin{subfigure}{0.32\textwidth}
    \centering
	 \includegraphics[width=\textwidth]{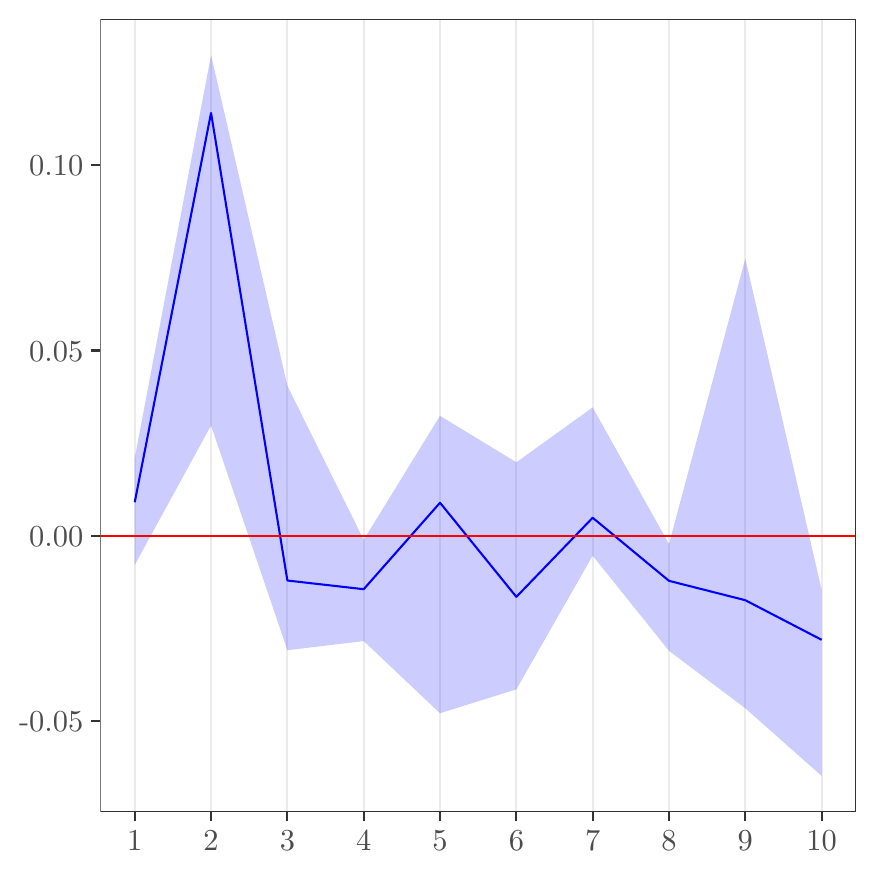}
	  \caption{NY Headquartered}
\end{subfigure}
\begin{subfigure}{0.32\textwidth}
    \centering
     \includegraphics[width=\textwidth]{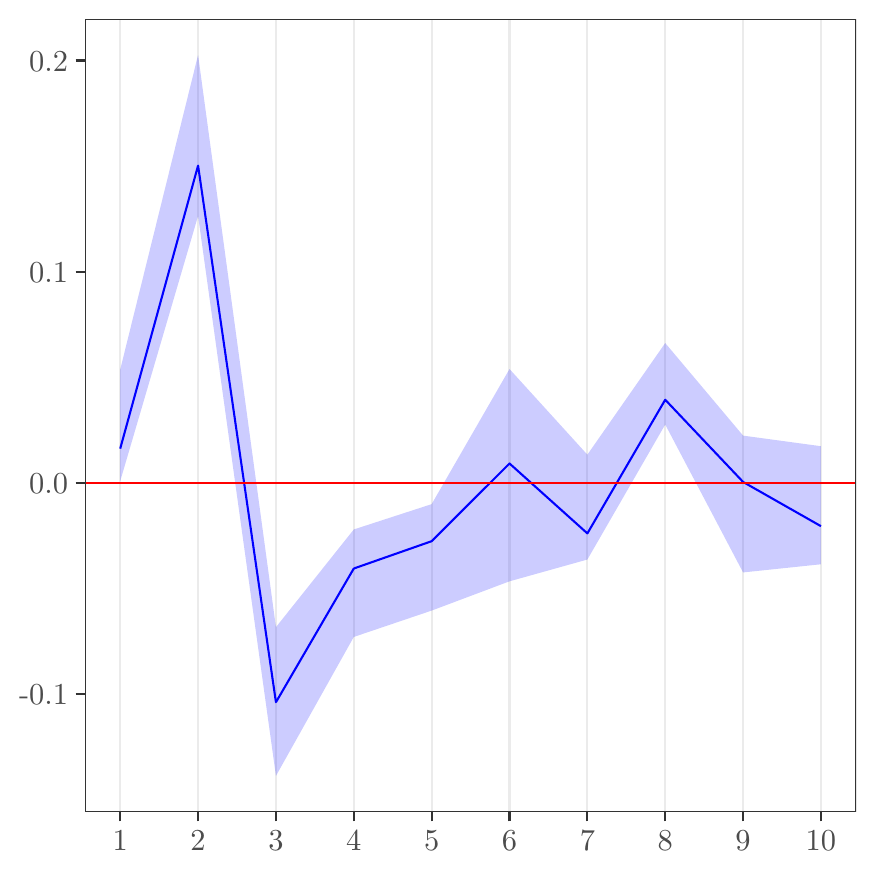}
    \caption{Personal Connection}
\end{subfigure}
\begin{subfigure}{0.32\textwidth}
    \centering
    \includegraphics[width=\textwidth]{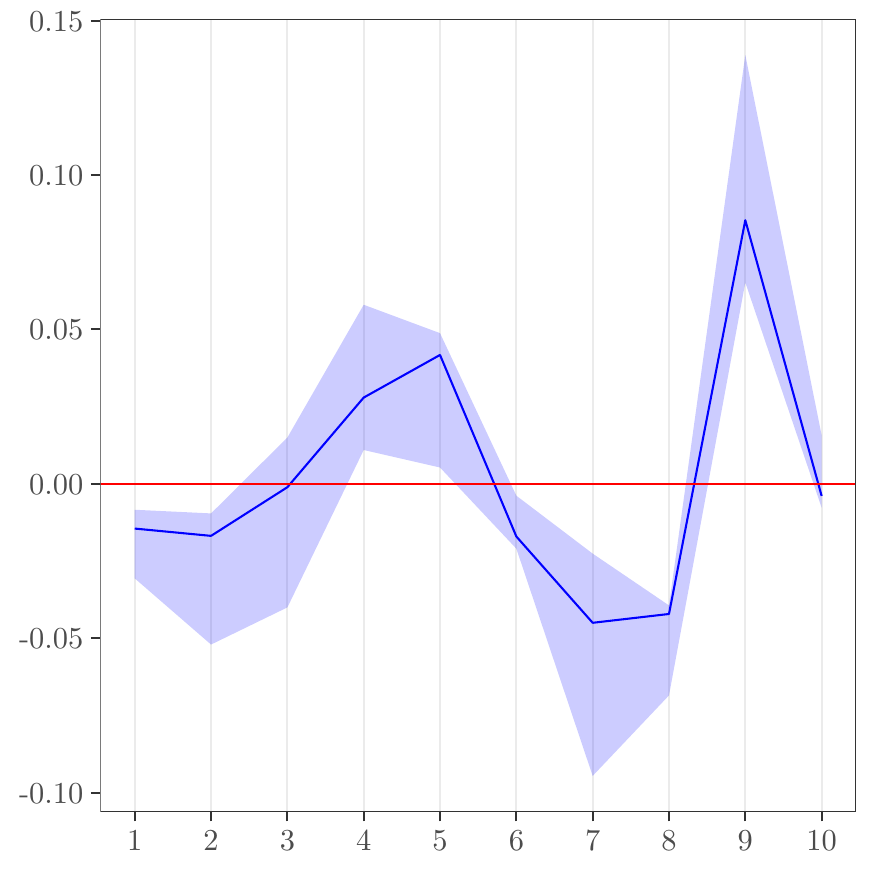}
	\caption{Schedule Connection}
\end{subfigure}
\begin{subfigure}{0.32\textwidth}
    \centering
	 \includegraphics[width=\textwidth]{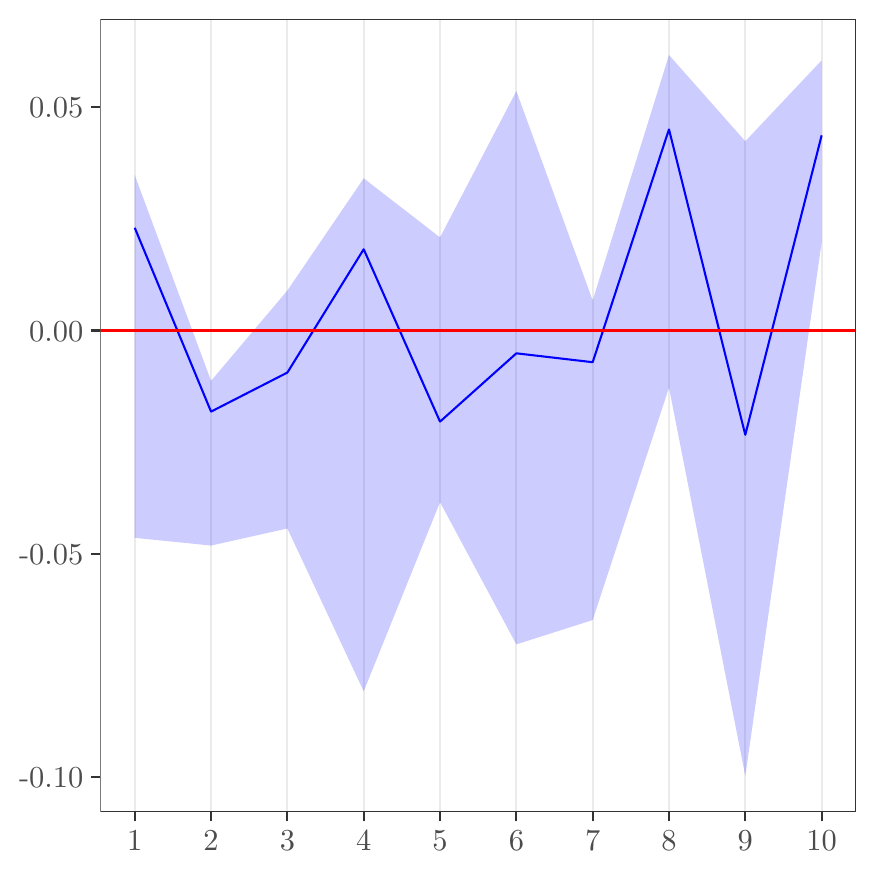}
	  \caption{NY Headquartered}
\end{subfigure}
\begin{subfigure}{0.32\textwidth}
    \centering
     \includegraphics[width=\textwidth]{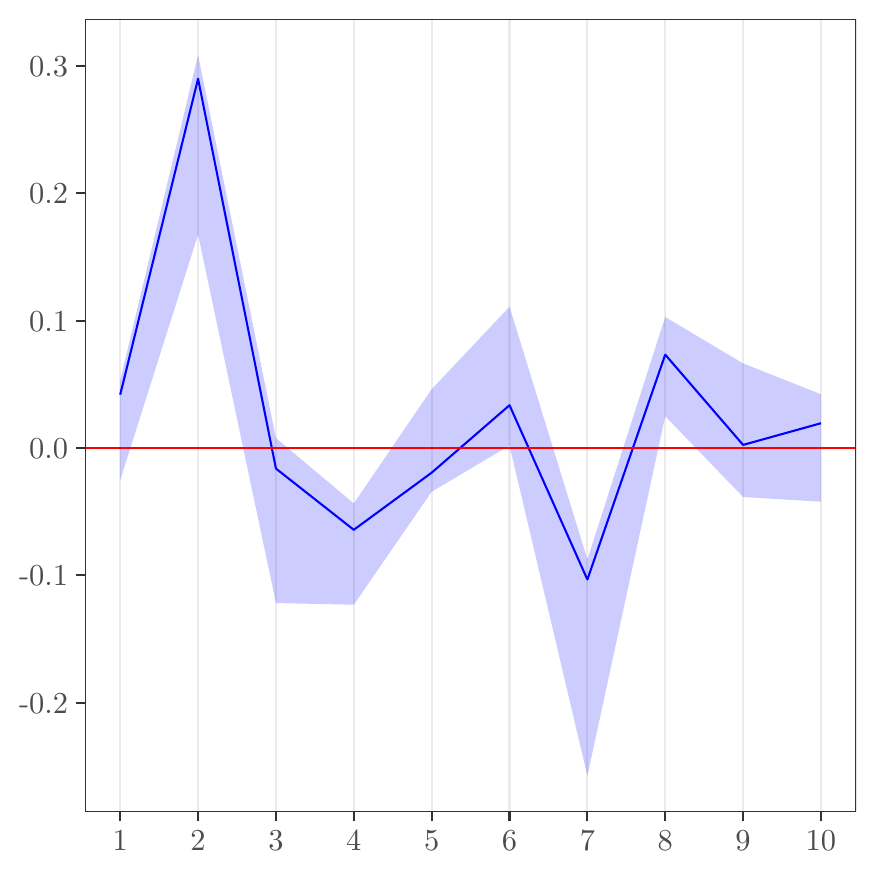}
    \caption{Personal Connection}
\end{subfigure}
\begin{subfigure}{0.32\textwidth}
    \centering
    \includegraphics[width=\textwidth]{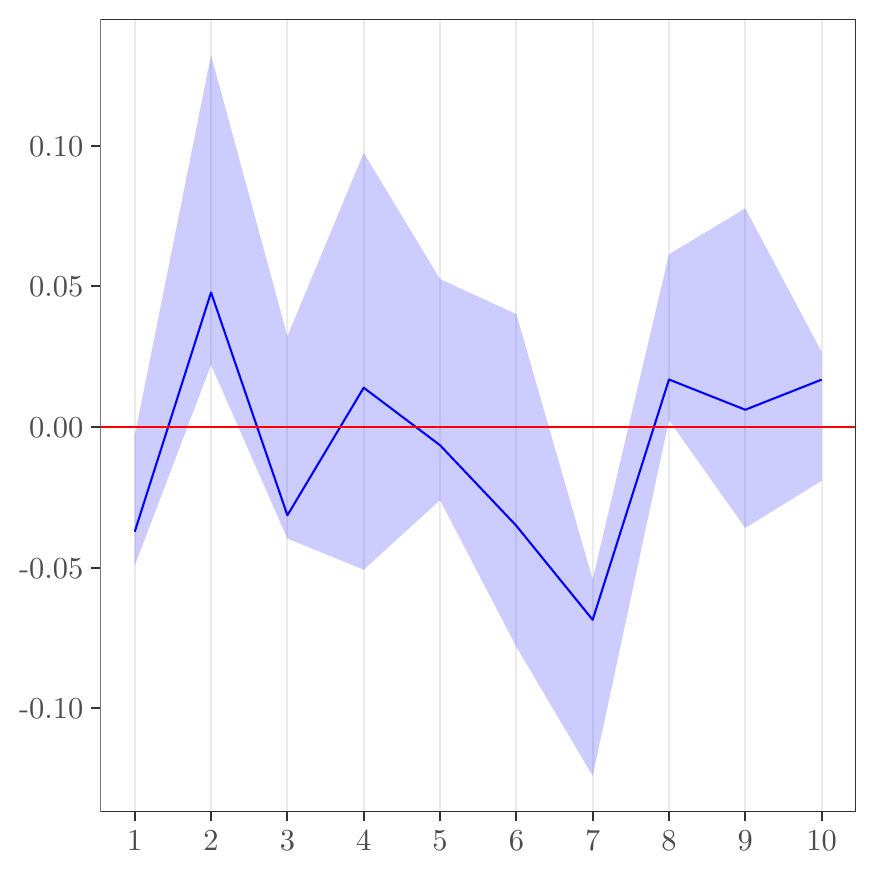}
	\caption{Schedule Connection}
\end{subfigure}
\begin{subfigure}{0.32\textwidth}
    \centering
	 \includegraphics[width=\textwidth]{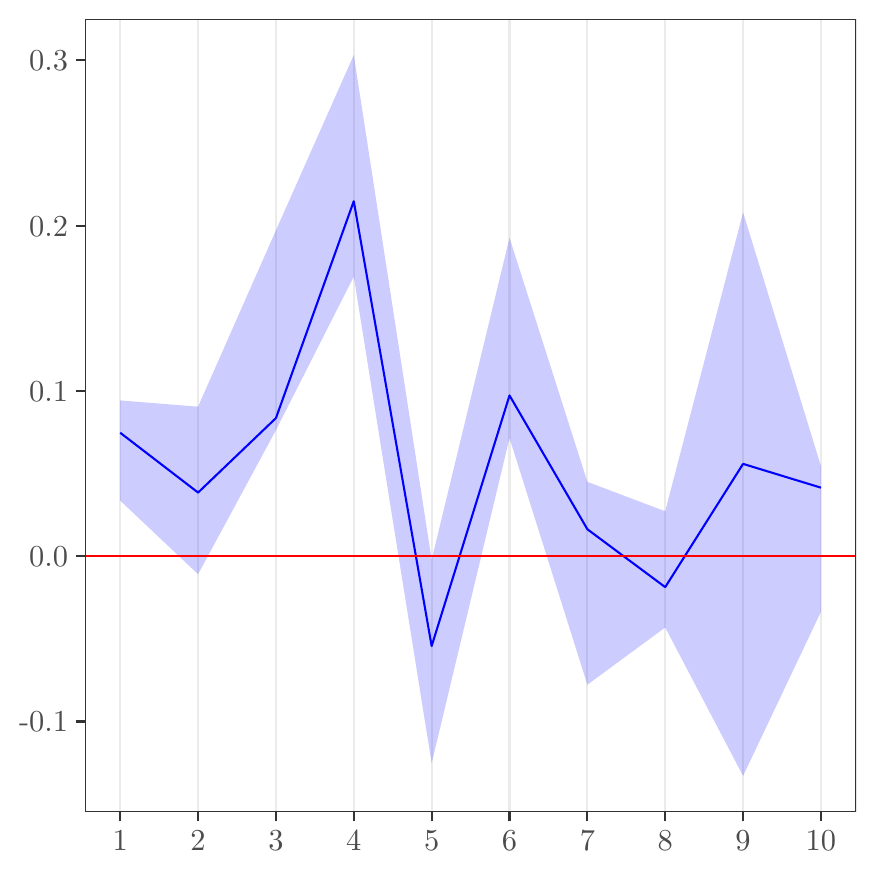}
	  \caption{NY Headquartered}
\end{subfigure}
\caption{Median and 95\% Confidence Interval for announcement effect on returns for each of the 10 post-announcement periods for the type of connection.}\label{fig:empirical_ci_appendix}
\end{figure}

\end{document}